\RequirePackage{etoolbox}
\csdef{input@path}{%
 {sty/}%
 {figures/}%
}%

\documentclass{article}

\usepackage{graphicx,amssymb,amsmath,amsthm,mathrsfs}
\usepackage{mathtools} 

\usepackage[titlenumbered,ruled]{algorithm2e}
\SetArgSty{textup}
\SetKwBlock{Loop}{loop}{end}
\SetKwComment{Comment}{$\triangleright$\ }{}
\SetCommentSty{textrm}

\usepackage{natbib}

\usepackage{hyperref}
\hypersetup{colorlinks,citecolor=blue,urlcolor=blue,filecolor=blue}

\usepackage[capitalise]{cleveref}
\crefname{assumption}{Assumption}{Assumptions}
\usepackage{autonum}

\newcommand{\aseq}[2]{a_{#1}^{(#2)}}
\newcommand{\ark}{\aseq{r}{k}}
\newcommand{\mseq}[2]{M_{#1}^{(#2)}}
\newcommand{\mrk}{\mseq{r}{k}}

\newcommand\R{{\mathbb R}} 
\newcommand\N{{\mathbb N}} 
\newcommand\Z{{\mathbb Z}} 
\newcommand{\statespace}{\mathscr{X}}
\newcommand{\D}{\mathcal{D}} 

\renewcommand\P{{\mathbb P}}        
\newcommand\E{{\mathbb E}}        
\newcommand\I{{\mathbf 1}}        
\newcommand{\var}{\operatorname{Var}}	

\newcommand{\ud}{\,\mathrm{d}}    

\newcommand\bracearraycond[1]{\left\{ \begin{array}{ll} #1 \end{array} \right.}

\def\iidsim{\stackrel{\scriptscriptstyle \textrm{iid}}{\sim}}         

\newtheorem{theorem}{Theorem}
\newtheorem{lemma}[theorem]{Lemma}
\newtheorem{proposition}[theorem]{Proposition}

\newtheoremstyle{noitalic}
{1em plus .2em minus .1em}
{1em plus .2em minus .1em}
{}  
{0pt}       
{\bfseries} 
{.}         
{5pt plus 1pt minus 1pt} 
{}          
\theoremstyle{noitalic}
\newtheorem*{remark*}{Remark}
\newtheorem*{conjecture*}{Conjecture}

\title{Pseudo-marginal inference for CTMCs on infinite spaces via monotonic likelihood approximations}

\author{Miguel Biron-Lattes\thanks{
    A. Bouchard-C{\^o}t{\'e} is supported by a National Sciences and Engineering Research Council of Canada (NSERC) Discovery Grant and Michael Smith Foundation COVID-19 Research Response. T. Campbell is supported by an NSERC Discovery Grant and an NSERC Discovery Launch Supplement.}\hspace{.2cm}\\
    \and 
    Alexandre Bouchard-C{\^o}t{\'e} \\
    \and
    Trevor Campbell \\
    Department of Statistics, University of British Columbia}

\begin{document}

\maketitle
\begin{abstract}
Bayesian inference for Continuous-Time Markov Chains (CTMCs) on countably
infinite spaces is notoriously difficult because evaluating the likelihood
exactly is intractable. One way to address this challenge is to first build
a non-negative and unbiased estimate of the likelihood---involving
the matrix exponential of finite truncations of the true rate matrix---and then
to use the estimates in a pseudo-marginal inference method. In this work, we 
show that we can dramatically increase the efficiency of this
approach by avoiding the computation of exact matrix exponentials. In particular, we 
develop a general methodology for constructing an unbiased, non-negative estimate of the likelihood
using doubly-monotone matrix exponential approximations. We further develop a novel 
approximation in this family---the skeletoid---as well as theory regarding its approximation error
and how that relates to the variance of the estimates used in pseudo-marginal inference.
Experimental results show that our approach yields more efficient posterior
inference for a wide variety of CTMCs.
\end{abstract}

\section{Introduction}

Continuous-Time Markov Chains (CTMCs) \citep{anderson1991continuous} are a
class of stochastic processes with piecewise constant paths,
taking values in a countable set $\statespace$, and switching between states
at random real-valued times.
Notable examples are the Stochastic
Susceptible-Infectious-Recovered (SSIR) model in epidemiology
\citep{keeling2008methods} and the Stochastic Lotka-Volterra (SLV)
predator-prey model in ecology \citep{spencer2008lotka}, but they 
have also been applied in
a wide range of other fields such as 
population genetics \citep{beerenwinkel2009markov}, engineering
\citep{yin2003simulation,lipsky2008queueing}, biochemistry
\citep{schaeffer2015} and phylogenetics \citep{maddison2007}. 
As an example,
consider the SSIR model in its simplest form. At any time $t\geq0$, the system
is described by a triplet of integers $(S(t), I(t), R(t))$, denoting the number
$S(t)$ of healthy people susceptible of being infected, the number $I(t)$ of
people infected, and the number $R(t)$ of recovered people who have become immune.
The system switches states at random times whenever a susceptible 
individual becomes infected, when an infected person recovers, or when a new 
susceptible individual enters the population.

Our interest lies in using observations to perform Bayesian statistical
inference for the unknown parameters of the CTMC model. In particular, we aim
to infer the \textit{rate matrix} $Q$ that governs the dynamics of the process.
For any pair of states $i\neq j$, $Q_{i,j}$ is related to the probability of
the system jumping from $i$ to $j$ in an ``infinitesimally small'' interval of
time. Evaluating the likelihood for any particular $Q$ involves computing the
\textit{matrix exponential} of $Q$, denoted $e^{tQ}$; in particular, the matrix
exponential lets us move from ``infinitesimal time intervals'' to transition
functions for time intervals of any length $t\geq0$ without approximation
error. However, we can only compute $e^{tQ}$ exactly when the system takes at
most a finite number of different states. And even in this case, the 
computation is feasible only when the size of $\statespace$ 
is relatively small due to the
$O(|\statespace|^3)$ complexity of the matrix exponential
\citep{moler_nineteen_2003}.  In many applications the set $\statespace$ where
the CTMC takes values is infinite, which in turn implies that $Q$ is
infinite-dimensional too.

One way to address this issue is to use the \textit{pseudo-marginal} method for
exact Bayesian inference \citep{andrieu2009pseudo}, where one is allowed to
replace exact evaluation of the likelihood by a noisy estimate of the
likelihood that is
\textit{non-negative} (i.e., guaranteed to be at least 0) and
\textit{unbiased} (i.e., with mean equal to the true likelihood).
\citet{georgoulas_unbiased_2017}  applied this strategy to a
subclass of CTMCs called \textit{Reaction Networks}, of which the SSIR
model is a particular example. The idea is to define a growing, infinite
sequence of finite sets $\statespace_1,\statespace_2,\statespace_3$, etc., that
eventually covers the entirety of $\statespace$. Then, a corresponding sequence
of finite rate matrices $Q_r$ is obtained, so that each $e^{tQ_r}$ can be
computed exactly. Finally, an unbiased estimate is obtained through a
randomization-based debiasing method
\citep{mcleish2011,rhee2015unbiased,lyne2015russian} that preserves
non-negativity.

The method described above can be recast as exploiting 
\emph{monotone approximations} to the transition probabilities of an infinite 
CTMC. Indeed, \citet{georgoulas_unbiased_2017} constructs a monotone 
approximation by computing matrix exponentials for rate matrices arising from 
increasing truncations of the state-space $\statespace$. However, we argue that 
exactly computing (expensive) matrix exponentials is wasteful. We show that this
can be avoided by building a monotone approximation that simultaneously 
increases the truncation of the underlying state-space as well as the accuracy 
of an approximation to the finite matrix exponential. By preserving monotonicity, 
this new sequence fits into the pseudo-marginal framework described earlier; 
and because it involves only approximate matrix exponentials, it improves 
computational efficiency. We demonstrate that these monotone approximations can
substantially improve the efficiency of pseudo-marginal methods for CTMCs.

Crucial to this new approach is the use of \emph{doubly-monotone approximations} to the
finite rate matrix exponential---i.e., those that are monotone non-decreasing
in \emph{both} state-space truncation size \emph{and} approximation iteration---because
they enable the construction of the monotone approximations to the transition probability
required by the pseudo-marginal method.
Past approximations applicable to this setting,
such as the well-known uniformization method \citep{jensen53}, are not doubly-monotone 
except in some limited circumstances. Therefore we propose the \emph{Skeletoid},
a novel matrix exponential approximation. We demonstrate that the skeletoid is 
doubly-monotone, and provide bounds on its approximation
error as a function of the number of iterations. 
Empirical results show that our pseudo-marginal method based on approximate matrix exponentials---using Skeletoid or uniformization---gives substantial efficiency improvements.

The remainder of the paper is organized as follows. \cref{sec:background}
introduces the transition function for a CTMC, a description of Reaction Networks,
and the general pseudo-marginal Markov chain Monte Carlo (MCMC) method.
\cref{subsec:nonneg_unbiased_est,subsec:designing_sequences_SS}
introduce the strategy to construct pseudo-marginal
samplers using monotone transition function approximations together with a
debiasing scheme. \cref{subsec:doubly_monotone} shows how to improve the
efficiency of this approach using doubly-monotone matrix exponential approximations.
\cref{subsec:monotone_approx_matrix_exp} introduces the
Skeletoid approximation, provides an analysis of its error, and compares it 
with uniformization.  Finally,
\cref{sec:experiments} provides a comparison of our proposed method and
alternative approaches. We conclude
with a discussion of avenues of improvement and generalizations to other
classes of stochastic processes.

\section{Background}\label{sec:background}

We begin this section by defining transition functions in the context of CTMCs,
and providing conditions under which the inference problem is well-posed. Then,
we define the subclass of CTMCs known as \emph{reaction networks}.
Finally, we give a general description of the pseudo-marginal approach to Bayesian inference.

\subsection{CTMCs, rate matrices, and transition functions}\label{subsec:transition_functions}

Let $\{X(t)\}_{t\geq 0}$ denote a continuous time Markov chain (CTMC) on a 
countable space $\statespace$ \citep[Ch. 8]{cinlar2013introduction}. The set of equations
\begin{equation}\label{eq:def_trans_func_as_prob}
M_{x,y}(t) := \P(X(t) = y | X(0) = x), \hspace{1em} x,y\in\statespace, t\geq 0,
\end{equation}
define the \emph{transition function of the process $\{X(t)\}_{t\geq 0}$}. In applications, the usual input available
to produce $M(t)$ is the \emph{rate matrix}
$Q = (q_{x,y})$ such that $q_{x,y} \ge 0$ for $x \neq y$. When $|\statespace|<\infty$, and $Q$ is \emph{conservative}, i.e.\ $q_{x,x} = - \sum_{y \neq x} q_{x,y}$, then $M(t)$ is given by the \emph{matrix exponential} of $Q$,
\begin{equation}\label{eq:def_matrix_exponential}
M(t) = e^{tQ} := \sum_{n = 0}^\infty \frac{(tQ)^n}{n!}, \hspace{1em} t\geq 0.
\end{equation}
One can use this relationship to compute the exact likelihood of observations. 
For countably infinite state-spaces $|\statespace|=\infty$, there is no such
closed-form expression relating $M(t)$ and $Q$ (they are
generally related via \emph{Kolmogorov's equations} \citep[][Thm. 2.2.2]{anderson1991continuous}). 
Therefore, specialized inference methods are required.
Note that the rate matrix $Q$ is still the inferential target in this setting, 
as it still determines the dynamics of the CTMC:
given a conservative and \emph{non-explosive} rate matrix $Q$  \citep[][Def. 8.3.22]{cinlar2013introduction},
Algorithm \ref{algo:gillespie} 
produces realizations of the CTMC that are distributed 
according to \cref{eq:def_trans_func_as_prob} \citep{gillespie1977exact}.

\begin{algorithm}[t]
\DontPrintSemicolon
\SetKwInOut{Input}{input}
\SetKwInOut{Output}{output}
\SetKw{Break}{break}
\Input{$x_0, Q, t_\text{end}>0$.}
$x\gets x_0; t\gets 0$\Comment*{initialize starting point}
$X \gets \{(x,t)\}$\Comment*{initialize path storage}
\Loop{
	$\tau \sim \text{Exp}(-q_{x,x}); t_\text{jump}\gets t+\tau$\Comment*{sample next jump time}
    \lIf(\Comment*[f]{reached the end}){$t_\text{jump} > t_\text{end}$}{\Break}
	$x \gets \text{Categorical}\left(\left\{\frac{q_{x,y}}{(-q_{x,x})} \right\}_{y\neq x}\right)$\Comment*{sample next state}
	$t\gets t_\text{jump}$\Comment*{advance time}
	$X \gets X\cup\{(x,t)\}$\Comment*{append new jump to path}
}
\Output{Path $X$.}
\caption{Gillespie's algorithm for sampling a path of a CTMC}
\label{algo:gillespie}
\end{algorithm}

\subsection{Reaction networks}\label{subsec:reaction_networks}

\emph{Reaction networks} are a class of CTMCs that admit a structured
parametrization of their rate matrices. The state space $\statespace$ of a reaction network 
is a subset of the integer lattice $\Z^{n_s}$, where $n_s$ is the
number of \textit{species} in the network. In most cases, the state spaces of
reaction networks can be described by (possibly infinite) rectangular sets
\[
\statespace = \{x\in\Z_+^{n_s}: B^l_i\leq x_i \leq B^u_i, \forall i\in\{1,\dots,n_s\}\}
\]
where $B_i^l,B_i^u\in\Z\cup\{-\infty,\infty\}$ define the lower and upper bounds on the counts of each species.

At any $x\in\statespace$, the process can only evolve by moving in a finite set
of directions in the lattice, specified by an integer valued  \textit{update matrix} $U$ of
size $n_r \times n_s$, where $n_r$ is the number of \textit{reactions}. Each
row in $U$ corresponds to an update vector, such that the next state is any one
of $\{x+U_{r,\cdot}\}_{r=1}^{n_r}$. The rate of jumps towards these states is
characterized by the \textit{propensity functions}
$\{f_\theta^r\}_{r=1}^{n_r}$, where $f_\theta^r:\statespace\to [0,\infty)$ and
$\theta\in\Theta$ is a parameter vector governing the system, on which we aim
to do inference. The $x$-th row of $Q$ is then given by
\begin{equation}\label{eq:reaction_network_rate_matrix}
Q_{x,y} =
\left\{
\begin{array}{ll}
f_\theta^r(x) & \text{if }y=x+U_{r,\cdot}\text{ for some }r,\\
-\sum_{r=1}^{n_r}f_\theta^r(x) & \text{if }y=x, \\
0 & \text{otherwise.}
\end{array}
\right.
\end{equation}
Note that this definition imposes a high level of sparsity on $Q$, as its rows
contain at most $n_r+1$ nonzero elements.

As mentioned earlier, the running SSIR model example is a reaction network with
$\statespace=\Z_+^3$; i.e., with $B^l=(0,0,0)$ and
$B^u=(\infty,\infty,\infty)$. It is represented by the following diagram
\begin{equation}\label{react:SIR}
\begin{array}{cccl}
S + I &\xrightarrow{\mathmakebox[2em]{\theta_1 SI}}& 2I &  \text{(infection)} \\
I &\xrightarrow{\theta_2 I} & R & \text{(recovery)} \\
\emptyset & \xrightarrow{\mathmakebox[2em]{\theta_3}} & S & \text{(arrival of susceptible)}
\end{array}
\end{equation}
where $\theta:=(\theta_1,\theta_2,\theta_3)$ are the model's parameters, and
the formulae above the arrows correspond to the propensity functions. The
update matrix induced by the diagram is
\begin{equation}\label{eq:SIR_update_matrix}
U=
\begin{pmatrix*}[r]
-1 &  1 & 0\\
 0 & -1 & 1\\
 1 &  0 & 0
\end{pmatrix*}.
\end{equation}

Most of the theory in this paper is applicable to general CTMCs, the only
exception being \cref{subsec:designing_sequences_SS}. However, we will focus
our applications and experiments on reaction networks because of their broad applicability in
scientific practice and the convenient structure they exhibit.

\subsection{Pseudo-marginal approach to inference}\label{subsec:pseudomar}

Consider a general statistical problem where one has collected a dataset $\D$ with the
intention of performing inference for an unknown parameter $\theta\in\Theta$.
The likelihood function is given by $L(\theta)$, and we put a prior
distribution on $\theta$ with density $\pi(\theta)$. The \textit{posterior
density} over $\theta$ is given by
\begin{equation}\label{eq:def_posterior}
\pi(\theta|\D) := \frac{\pi(\theta)L(\theta)}{p(\D)} = \frac{\pi(\theta)L(\theta)}{\int_\Theta\pi(\theta')L(\theta')\mathrm{d}\theta'}.
\end{equation}
When $\pi(\theta)$ and $L(\theta)$ can be evaluated point-wise, standard tools
such as Markov Chain Monte Carlo (MCMC) can be used to carry out Bayesian
inference without the need to calculate $p(\D)$. However, in many situations
the likelihood $L(\theta)$ cannot be evaluated, either because it does not
admit a closed form expression, or because using such an expression requires
infinite computational resources.

The pseudo-marginal approach \citep{andrieu2009pseudo} is a useful tool to approach cases where $L(\theta)$ is intractable. Suppose that we have access to a non-negative estimator of the
likelihood $\tilde{L}(\theta,U)\geq0$, where $U$ is an auxiliary  $\mathcal{U}$-valued 
independent random variable with density $m(u)$, such that
$\E[\tilde{L}(\theta,U)]=L(\theta)$. Define the \textit{augmented posterior distribution}
\begin{equation}\label{eq:pseudo_marginal_aug_target_density}
\pi(\theta, u|\D) := \pi(\theta|\D) m(u) \frac{\tilde{L}(\theta,u)}{L(\theta)} \propto \pi(\theta) m(u) \tilde{L}(\theta,u) =: \gamma(\theta,u).
\end{equation}
We can check that this augmented density has the
correct marginal for $\theta$, since
\[
\int_\mathcal{U} \pi(\theta, u|\D) \mathrm{d}u = \frac{\pi(\theta|\D)}{L(\theta)} \int_\mathcal{U} \tilde{L}(\theta,u) m(u)\mathrm{d}u = \frac{\pi(\theta|\D)}{L(\theta)}\E[\tilde{L}(\theta,U)] = \pi(\theta|\D).
\]
Moreover, $\gamma(\theta,u)$ can be evaluated point-wise because it does not explicitly
depend on $L(\theta)$. Therefore, we can perform  inference for $\theta$
by employing any suitable MCMC algorithm targeting $\gamma$.

\section{Inference in countably infinite state spaces}\label{sec:inference}

Suppose that we observe a CTMC at a finite set of times
$\{t_i\}_{i=0}^{n_d}$ with states $\{x_i\}_{i=0}^{n_d}$. Denote this data by
$\D=\{x_i,t_i\}_{i=0}^{n_d}$. We aim to perform inference for the parameters
$\theta \in \Theta$ that define the rate matrix $Q=Q(\theta)$. The
likelihood has the form
\begin{equation}
L(\theta) := \prod_{i=1}^{n_d} M(t_i - t_{i-1};\theta)_{x_{i-1}, x_i},
\end{equation}
where $M(\delta;\theta)_{x,y}$ is the transition probability from state $x$ to
$y$ after a time interval of length $\delta$ for the CTMC governed by the rate matrix $Q(\theta)$
(we write $M(t)$ instead of $M(t;\theta)$ to simplify notation). As we have previously discussed, for countably infinite spaces
there is no general exact method to compute $M(u;\theta)_{x,y}$ in finite time.
In this section we describe a novel approach that uses the pseudo-marginal method to
address this problem.

\subsection{Non-negative unbiased estimators of the likelihood}\label{subsec:nonneg_unbiased_est}

To construct an unbiased estimator of the likelihood $L(\theta)$, we design a
debiasing scheme
\citep{kahn1954applications,mcleish2011,rhee2015unbiased,lyne2015russian,jacob2015nonnegative}
that exploits the monotonicity of increasing sequences converging to a quantity
of interest (all proofs in \cref{app:proofs}).

\begin{proposition}\label{prop:debiasing}
Let $\{a_n\}_{n\in\Z_+}$ be a monotone increasing sequence 
with $a_n \uparrow \alpha <\infty$. Fix $\omega \in \Z_+$ and let $N\in\Z_+$ be an independent
random variable with probability mass function $p(n)$ satisfying
\begin{equation}\label{eq:debias_estimator_assu_domination}
\forall n \in \Z_+ : \,\, p(n) = 0, \qquad a_{\omega+n+1} = a_{\omega+n}.
\end{equation}
Define the random variable
\begin{equation}\label{eq:debias_estimator}
Z := a_{\omega} + \frac{a_{\omega+N+1} - a_{\omega+N}}{p(N)}.
\end{equation}
Then $Z \geq a_\omega$ almost surely, $\E[Z]=\alpha$, and
\[
 \var(Z) = \sum_{n:p(n)>0} \frac{(a_{\omega+n+1} - a_{\omega+n})^2}{p(n)} - (\alpha-a_\omega)^2. \label{eq:debias_variance}
\]
Furthermore, if $\var(Z) < \infty$ for a fixed $\omega\in\Z_+$, then $\var(Z) \to 0$ as $\omega \to \infty$.
\end{proposition}

We refer to the above method as the Offset Single Term Estimator (OSTE) because
it builds on the single-term estimator of \citet{rhee2015unbiased} by
incorporating the offset $\omega$ as a tuning parameter. OSTE has two advantages over
other debiasing schemes. First, it requires computing at most $3$ (random)
elements of the sequence $\{a_n\}_{n\in\Z_+}$, and only $2$ with probability
$\P(N=0)$. This provides a large reduction in computation in situations---as in the present work---where 
$a_n$ can be evaluated
directly without needing to first evaluate $a_1, \dots, a_{n-1}$. Second, 
the offset allows us to control a trade-off between computation cost and variance separate
from the design of the distribution $p$. 

In the context of pseudo-marginal inference, it is important
to control the variance of the estimate, as it directly influences the convergence rate of the sampler. Aside from the offset,
\cref{eq:debias_variance} suggests that the variance of the OSTE is determined
by the relative values of the difference sequence $a_{n+1}-a_n$ and the probability
mass function $p(n)$. The faster $a_n$ converges, the lighter-tailed $p(n)$ can be;
in situations where $a_n$ becomes more expensive to compute as $n$ increases---as in the present work---it 
is desirable to use a light-tailed $p(n)$ that rarely returns large values of $N$.
In particular, a case of special interest to the present work is when the sequence $\{a_n\}_{n\in\Z_+}$ 
converges exponentially fast to its limit. 
\cref{prop:suff_cond_debiasing} shows that in this case, we can use a geometric
distribution $p$ while still guaranteeing that $\var(Z) < \infty$.
\begin{proposition}\label{prop:suff_cond_debiasing}
Suppose there exist constants $c, \rho>0$ such that
\[
\forall n\in\Z_+, \qquad a_{n+1}-a_n \leq ce^{-n\rho}. \label{eq:debias_assu_exp_error_decay}
\]
Then if $N \sim \mathrm{Geom}\left(1-e^{-\beta}\right)$ for $0 < \beta < 2\rho$,
\[
\var(Z) \leq \frac{c^2 e^{-2\omega \rho}}{(1-e^{-\beta})(1-e^{\beta-2\rho})}.
\]
\end{proposition}
In this work we follow this general strategy of
creating monotone sequences that converge at least exponentially quickly 
to their limit. In practice, the constants $c, \rho$ above are 
typically unknown; in \cref{app:tuning_stopping_times} we describe a 
procedure to automatically tune OSTE that is suitable for the present work.

It remains to design an appropriate monotone sequence that converges to
$L(\theta)$. A first step in this direction is given by the following result
due to \cite{anderson1991continuous}, which shows that any increasing sequence
of state spaces induces an increasing sequence of transition functions.

\begin{proposition}[{\citet[Prop.~2.2.14]{anderson1991continuous}}]\label{prop:monotone_trunc}
Let $\{\statespace_r\}_{r\in \Z_+}$ be an increasing sequence of sets such that $\bigcup_{r\geq0}\statespace_r = \statespace$. Define a corresponding sequence of rate matrices $\{Q_r\}_{r\in \Z_+}$ by
\begin{equation}\label{eq:def_truncated_Q}
(Q_r)_{x,y} :=
\left\{
\begin{array}{cc}
Q_{x,y} & x,y\in \statespace_r \\
0 & \text{o.w.}
\end{array}
\right.
\end{equation}
Then for all $x,y\in\statespace$ and $t\geq 0$, $[e^{tQ_r}]_{x,y}\uparrow M(t)_{x,y}$ as $r\to\infty$.
\end{proposition}

\cref{prop:monotone_trunc} tells us how to construct an increasing sequence of
probabilities which we can pass to OSTE to produce a non-negative and unbiased
estimator of the transition probability of a CTMC. Note that the truncated rate matrices $Q_r$ are not guaranteed to be \emph{conservative}, meaning that some of its rows may fail to sum to a value larger than zero (due to the missing states). We refer to these matrices as \emph{non-conservative}. 
Nonetheless, \cref{prop:matexp_nonconservative_rate_mat} and its proof  
shows that $[e^{tQ_r}]_{x,y}$ can be interpreted as a transition probability in an enlarged space.  

Let us assume that for each
observation $(x_{i-1},x_i,\Delta t_i)$ we have constructed an increasing
sequence of state-spaces $\{\statespace_r^{i}\}_{r\in\Z_+}$. Fix
$\omega_i\in\Z_+$ and let $N_i$ be an independent random integer with
probability mass function $p_i(n)$. Then an estimator of 
$L_i(\theta):=M(\Delta t_i)_{x_{i-1},x_i}$ is given by applying the OSTE debiasing scheme to the
monotone sequence of transition probabilities 
$\{M_r(\Delta t_i)_{x_{i-1},x_i}\}_{r\in\Z_+}$ given by \cref{prop:monotone_trunc},
\begin{equation}\label{eq:def_est_IE_elements}
\tilde{L}_i^{\text{IE}}(\theta,N_i) := M_{\omega_i}(\Delta t_i)_{x_{i-1},x_i} + \frac{M_{\omega_i+N_i+1}(\Delta t_i)_{x_{i-1},x_i}-M_{\omega_i+N_i}(\Delta t_i)_{x_{i-1},x_i}}{p_i(N_i)}.
\end{equation}
IE stands for an estimator for \textbf{I}rregular time series based on
\textbf{E}xact matrix exponentials. When the increasing sequence of state-spaces
is designed such that the conditions of
\cref{prop:debiasing} are satisfied, $\tilde{L}_i^{\text{IE}}(\theta,N_i)$ is
unbiased. Furthermore, by the independence of $\{N_i\}_{i=1}^{n_d}$, 
we can obtain an unbiased estimator of the likelihood via
\begin{equation}\label{eq:def_est_IE}
\tilde{L}_{\text{IE}}(\theta,U) := \prod_{i=1}^{n_d} \tilde{L}_i^{\text{IE}}(\theta,N_i),
\end{equation}
where $U:=(N_1,\dots,N_{n_d})$.

As the emphasis on irregular time series suggests, there is another estimator
that is suitable only for regular time series; i.e., with $\Delta t_i =\Delta t$ 
for all $i\in\{1,\dots,n_d\}$. It is motivated by a potential efficiency
gain associated with the fact that in this case we can obtain all the
transition probabilities $\{M(\Delta t)_{x_{i-1},x_i}\}_{i=1}^{n_d}$ from the
same matrix exponential $M(\Delta t)$. To exploit this property we begin by
collapsing the collection of sequences of state-spaces
$\{\statespace_r^i\}_{i,r}$ into a unique sequence
$\{\statespace_r\}_{r\in\Z_+}$
\begin{equation}\label{eq:def_RE_merged_seq_SS}
\statespace_r := \bigcup_{i=1}^{n_d} \statespace_r^i.
\end{equation}
Since each sequence $\{\statespace_r^i\}_{r\in\Z_+}$ is increasing to
$\statespace$, the collapsed sequence also satisfies this property. Then we
proceed as before: fix $\omega\in\Z_+$ and let $N$ be a random integer with
probability mass function $p(n)$. Define
\begin{equation}\label{eq:def_est_RE_L_hat}
\tilde{L}_r^{\text{RE}}(\theta) := \prod_{i=1}^{n_d} M_r(\Delta t)_{x_{i-1},x_i}.
\end{equation}
Applying the OSTE debiasing scheme to the sequence
$\{\tilde{L}_r^{\text{RE}}(\theta)\}_{r\in\Z_+}$ yields a new estimator of the
likelihood
\begin{equation}\label{eq:def_est_RE}
\tilde{L}_{\text{RE}}(\theta,N) := \tilde{L}_\omega^{\text{RE}}(\theta) + \frac{\tilde{L}_{\omega+N+1}^{\text{RE}}(\theta) - \tilde{L}_{\omega+N}^{\text{RE}}(\theta)}{p(N)}
\end{equation}
RE now stands for ``\textbf{R}egular time series with \textbf{E}xact
matrix exponentials.'' The estimator $\tilde{L}_{\text{RE}}(\theta,N)$ is unbiased whenever the
increasing state-space sequence is designed such that the
conditions of \cref{prop:debiasing} hold.

The strategy presented in this section is not restricted in
principle to CTMCs. In other words, the problem of performing exact inference
on stochastic processes with intractable transition functions can be reduced to
constructing monotone approximations to this function. Then, passing this
sequence through a debiasing scheme such as OSTE produces a likelihood
estimator which can be readily used as the basis of a pseudo-marginal sampler
targeting the correct (augmented) posterior distribution.

\subsection{Designing the state-space sequence}\label{subsec:designing_sequences_SS}

Up until this point, we have assumed the existence of a sequence of state spaces
increasing to $\statespace$. In this section we describe a procedure to create
such a sequence that leverages the structure of reaction networks.

For every observation $(x_{i-1},x_{i},\Delta t_i)$, we set $\statespace_0^i$ 
to be a positive probability path between $x_{i-1}$ and $x_i$. Then, we 
form an increasing sequence of sets using a simple iterative rule.
Let $D$ be any collection of vectors that span $\statespace$. To grow
the current space $\statespace_r^i$, we get new points by ``moving'' all the
elements of $\statespace_r^i$ one step along every direction in $D$, and then
discarding points outside $\statespace$
\begin{equation}\label{eq:state_space_seq_grow_rule}
\statespace_{r+1}^i = \statespace_r^i\cup(\{s+d:s\in\statespace_r^i,d\in D\}\cap \statespace).
\end{equation}
In this work we set $D$ to the canonical basis vectors of $\Z^{n_s}$ and 
their negations $D:=\{\pm e_j\}_{j=1}^{n_s}$; this is  
the default choice in the literature \citep{georgoulas_unbiased_2017,sherlock_exact_2019}. 
Note that this choice implies that the truncated state space
$\statespace_r^i$ grows polynomially, i.e., $|\statespace_r^i|=O(r^{n_s})$.
Another possible choice of $D$, for example, is to use the rows of the update matrix $U$ of the CTMC (e.g., 
\cref{eq:SIR_update_matrix} for the SIR model).

The benefit of including a valid path between endpoints in $\statespace_0^i$ is
that it ensures that the transition probability $M_0(\Delta t_i)_{x_{i-1},x_i}$ 
estimated with the corresponding truncated rate matrix $Q_0^{(i)}$ is positive.
When this condition fails, the performance of the pseudo-marginal sampler is
impaired because proposals with zero estimated likelihood are immediately
rejected. 
In \cref{app:build_base_statespace} we describe a simple heuristic based on 
linear programming to obtain $\statespace_0^i$. 
More general-purpose
pathfinding algorithms like $A^*$ \citep{hart1968formal} could also be used. 
Regardless of the algorithm used, one only
needs to build $\statespace_0$ once before running the sampler.

We now investigate the convergence rate of the estimated transition probability
sequence using this state-space truncation scheme. The goal is to show that
the error decays exponentially quickly, and hence that we may use a light-tailed
random truncation levels $N_i$ by \cref{prop:suff_cond_debiasing}.
Let $(x,y,t)$ represent a generic observation, and let $\P_x$ be the law of the
CTMC initialized at $X(0)=x$. As \cref{prop:matexp_nonconservative_rate_mat} shows, the $r$-th estimated transition
probability $M_r(t)_{x,y}$ can be written as the probability of moving from $x$
to $y$ without ever leaving $\statespace_{r}$, 
\[
M_r(t)_{x,y} = \P_x(X(t)=y\text{ and }\forall s \in [0,t]: X(s)\in\statespace_{r}).
\] 
Then
\begin{align}
M(t)_{x,y} - M_r(t)_{x,y} &= \P_x(X(t)=y) -  \P_x(X(t)=y\text{ and }\forall s \in [0,t], \, X(s)\in\statespace_{r}) \\
&= \P_x(X(t)=y\text{ and }\exists u\in[0,t] : X(u)\notin\statespace_{r}) \\
&\leq \P_x(\exists u\in[0,t]: X(u)\notin \statespace_{r}).
\end{align}
Reaching a point
$y\in\statespace\setminus\statespace_{r}$ 
from a point $x\in\statespace_0$
involves at the very minimum a
trip from the boundary of $\statespace_0$ to the boundary of $\statespace_r$. 
Since the directions $D$ span $\statespace$, there exists a constant $\gamma>0$ such that 
this trip requires at least $\gamma r$ jumps.
Therefore
\begin{align}
M(t)_{x,y} - M_r(t)_{x,y} &\leq \P_x(N(0,t]> \gamma r),
\end{align}
where $N(0,t]$ is the number of jumps the process makes in $(0,t]$.
To our knowledge, there are no analytically tractable expressions for tail
probabilities of $N(0,t]$ for general, non-explosive CTMCs in countably infinite state spaces.
Nevertheless, when the full rate matrix $Q$ is \textit{uniformizable}, i.e., 
\[
\bar{q}:=\inf_{x\in \statespace}q_{x,x} > -\infty,
\]
it is possible to derive an upper bound. In particular, one can add 
self-transitions to the CTMC---a process that is known as
\textit{uniformization} \citep{jensen53}---so that the enlarged set of jump
events is given by a homogeneous Poisson process $\bar{N}$ with intensity $-\bar{q}$.
Then, standard results for the Poisson distribution \citep[see e.g.][]{boucheron2013concentration} yield
\[
M(t)_{x,y} - M_r(t)_{x,y}\leq \P(N(0,t]> \gamma r) \leq \P(\bar{N}(0,t]> \gamma r) = O(e^{-\gamma r\log(\gamma r)}).
\]
Thus, under the assumption of uniformizability, the error in the truncation 
sequence decreases superexponentially. This is desirable given our aim to use the OSTE scheme, because it implies 
that \cref{eq:debias_assu_exp_error_decay} in \cref{prop:suff_cond_debiasing}
is satisfied.

\begin{figure}[t]
\centering
\includegraphics[width=\textwidth]{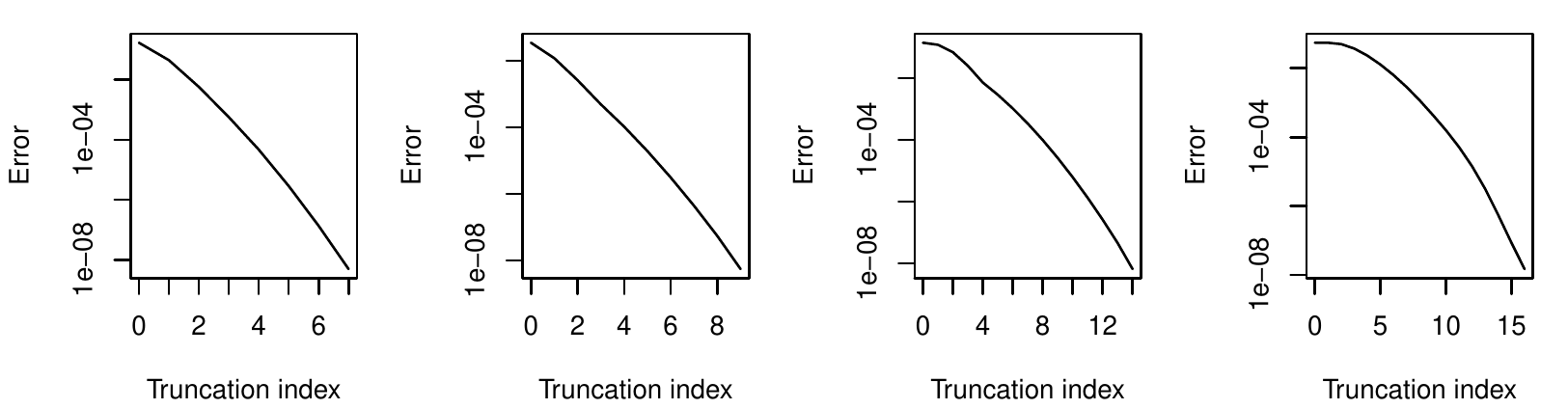}
\includegraphics[width=\textwidth]{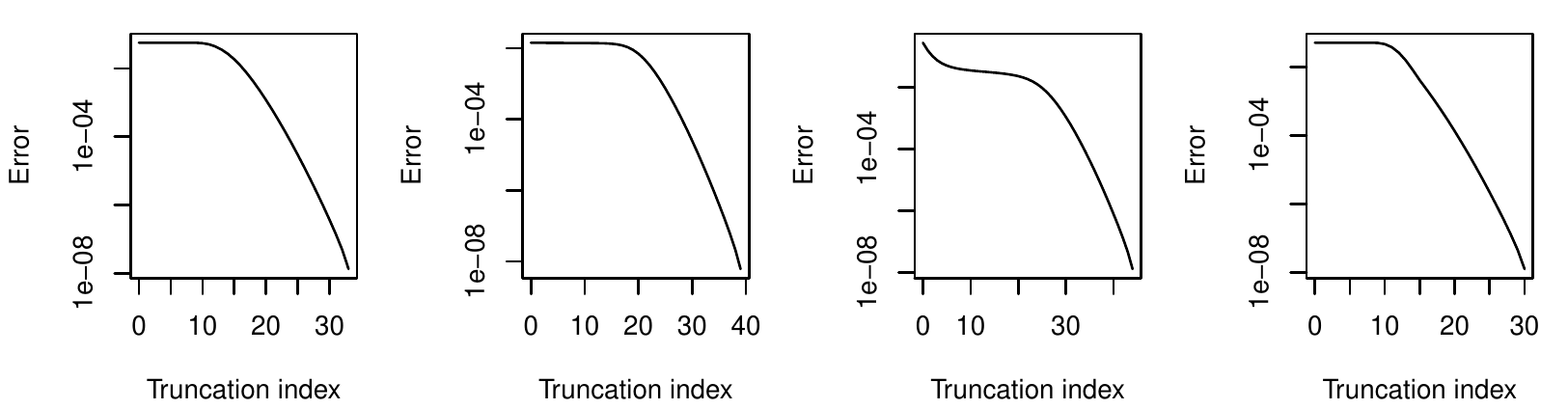}
\caption{Error in estimated transition probabilities versus truncation index $r$ for selected transition pairs. Within each row, a column corresponds to an observation $(x_i,y_i,\Delta t_i)$. The error is computed here as $M(\Delta t_i)_{x_i,y_i} - M_r(\Delta t_i)_{x_i,y_i}$. \textbf{Top:} $M/M/c$ model. \textbf{Bottom:} Schl{\"o}gl model.}
\label{fig:truncation_conv}
\end{figure}

\cref{fig:truncation_conv} shows the error incurred versus the truncation index $r$ 
for selected observation pairs in two CTMCs. The top row corresponds to a
simplified model of a queue. It assumes jobs arrive at a rate $\lambda$, and
that there are $c$ servers available to process those jobs at a rate of $\mu$.
This model---known as an $M/M/c$ queue \citep{lipsky2008queueing}---satisfies 
the uniformization condition, since $\bar{q} = -(\lambda+c\mu)>-\infty$. 
The plots show that the decrease in error is slightly
faster than exponential, in agreement with the derivation above. The bottom row
in \cref{fig:truncation_conv} corresponds to the Schl{\"o}gl model, another
CTMC which will be described in detail in \cref{sec:experiments}. In contrast
to the $M/M/c$ queue, the Schl{\"o}gl model has a rate matrix that is \emph{not}
uniformizable. Nevertheless, the observed decay in error still seems
superexponential, which suggests that fast convergence is robust to violations of
the uniformization assumption.

In practical application, the specific constants governing the rate of (super)exponential
error decay are not known. In \cref{app:tuning_stopping_times} we show how to empirically
estimate these rates as part of the process of tuning a pseudo-marginal sampler.

\subsection{Avoiding exact computation of matrix exponentials}\label{subsec:doubly_monotone}

The IE and RE estimators described in \cref{subsec:nonneg_unbiased_est} require
algorithms that produce an exact  evaluation of
$e^{tQ_r}$ for each truncated rate matrix $Q_r$. In this section, we show how to
construct new unbiased likelihood estimators that take advantage of monotonic
approximations of each $e^{tQ_r}$ to avoid the cost of exact computation. 
In particular, suppose for each $r\in\Z_+$ we
are given a monotone approximation $M_r^{(k)}(t)$, $k\in\Z_+$ to $M_r(t)$, $M_r^{(k)}(t) \uparrow M_r(t)$ entry-wise as $k\to\infty$ (we describe such algorithms in 
\cref{subsec:monotone_approx_matrix_exp}).
We say that the collection $\{M_r^{(k)}(t)\}_{r,k\in\Z_+}$ is \emph{doubly-monotone}
if it also satisfies 
\[
\forall r, k\in\Z_+, \qquad M_r^{(k)}(t)\leq M_{r+1}^{(k)}(t).
\]
The following result shows that for a doubly-monotone collection $\{M_r^{(k)}(t)\}_{r,k\in\Z_+}$, 
any sequence of matrices constructed from pairs of indices
$\{(r_n,k_n)\}_{n\in\Z_+}$ with $r_n \uparrow \infty$ and $k_n \uparrow \infty$
converges monotonically to the correct limit.
\begin{proposition}\label{prop:doubly_monotone_sequences}
Consider a collection 
\begin{equation}\label{eq:def_ark_collection}
\mathcal{A} = \{\ark:r\in\Z_+,k\in\Z_+\}\subset \R.
\end{equation}
Suppose that 
for all $r\in\Z_+$, $\ark \uparrow a_r\in\R$ where 
 $a_r\uparrow \alpha\in\R$,
and that 
for all $k\in\Z_+$, $\{\ark\}_{r\in\Z_+}$ is monotone increasing.
Let $\{(r_n, k_n)\}_{n\in\Z_+}$ be any sequence of pairs
$r_n, k_n\in\Z_+$ such that $r_n \uparrow \infty$ 
and $k_n \uparrow \infty$. Then $\aseq{r_n}{k_n} \uparrow \alpha$.  
\end{proposition}

Applying \cref{prop:doubly_monotone_sequences} to the collection
\begin{equation}\label{eq:def_collection_approx_trans_prob}
\mathcal{A} := \{M_r^{(k)}(t)_{x,y}: r\in\Z_+, k\in\Z_+\}
\end{equation}
for each matrix entry $(x,y)$
ensures that we can pick any sequence of increasing pairs $\{(r_n,k_n)\}_{n\in\Z_+}$ 
and obtain a sequence $M_{r_n}^{(k_n)}(t)$ that converges monotonically to the true transition probability.
Therefore, we can construct a debiased estimator using the result of \cref{prop:debiasing}
without ever having to evaluate the exact matrix exponential $M_r(t)$ of any state-space truncation. 
In order to control the variance of the estimator via \cref{prop:suff_cond_debiasing}, we require 
further conditions on the approximation sequence given by \cref{prop:doublymonotonevariance}.

\begin{proposition}\label{prop:doublymonotonevariance}
In addition to the conditions in \cref{prop:doubly_monotone_sequences}, suppose that there exists $c_1,c_2,\rho,\kappa>0$ such that
\begin{equation}\label{eq:doublymonotonevariance_assumptions}
\forall r\in\Z_+: \alpha-a_r \leq c_1e^{-\rho r}, \hspace{1em} \text{ and } \hspace{1em} \forall r,k\in\Z_+: 
a_r - \ark \leq c_2e^{-\kappa k}.
\end{equation}
Then, setting $r_n=n$ and $k_n=\lceil(\rho/\kappa) n\rceil$ gives for all $n\in\Z_+$
\begin{equation}\label{eq:doublymonotonevariance_result}
\aseq{r_{n+1}}{k_{n+1}} - \aseq{r_n}{k_n} \leq 2\max\{c_1,c_2\}e^{-\rho n}.
\end{equation}
\end{proposition}

The first assumption in \cref{eq:doublymonotonevariance_assumptions} is discussed in \cref{subsec:designing_sequences_SS}. The second condition is satisfied whenever the matrix exponential approximation admits an upper bound on the error that is invertible (\cref{subsec:monotone_approx_matrix_exp} gives such bounds for the methods described therein). Finally, we again note that the constants in \cref{prop:doublymonotonevariance} are unknown in practice, but nonetheless the sequences $r_n,k_n$ can still be tuned as we show in \cref{app:tuning_stopping_times}.

We now proceed as in
\cref{subsec:nonneg_unbiased_est} to define a new estimator of the true
transition probability $M(\Delta t_i)_{x_{i-1},x_i}$ for the $i$-th
observation. Suppose $\{M_r^{(k)}\}_{r,k\in\Z_+}$ is doubly-monotone,
and let $\{r_i(n),k_i(n)\}_{n\in\Z_+}$ be a sequence of increasing pairs.
Then
\begin{equation}\label{eq:def_est_IA_elements}
\tilde{L}_i^{\text{IA}}(\theta,N_i) := M_{r_i(0)}^{(k_i(0))}(\Delta t_i)_{x_{i-1},x_i} + \frac{M_{r_i(N+1)}^{(k_i(N+1))}(\Delta t_i)_{x_{i-1},x_i} - M_{r_i(N)}^{(k_i(N))}(\Delta t_i)_{x_{i-1},x_i}}{p_i(N_i)}
\end{equation}
is an unbiased estimator of $M(\Delta t_i)_{x_{i-1},x_i}$. In turn, these estimators can then be composed into an unbiased estimator of the likelihood
\begin{equation}\label{eq:def_est_IA}
\tilde{L}_{\text{IA}}(\theta,U) := \prod_{i=1}^{n_d} \tilde{L}_i^{\text{IA}}(\theta,N_i),
\end{equation}
where $U:=(N_1,\dots,N_{n_d})$. We use the suffix IA to denote the use of
\textbf{A}pproximate solutions in the setting of \textbf{I}rregular
time-series.
Likewise, for regular time series let
\begin{equation}\label{eq:def_joint_approx_like}
L_r^{(k)}(\theta) := \prod_{i=1}^{n_d} \mrk(\Delta t)_{x_{i-1},x_i}.
\end{equation}
Applying \cref{prop:doubly_monotone_sequences} to $\mathcal{A}=\{L_r^{(k)}(\theta):r\in\Z_+,k\in\Z_+\}$ shows 
that
\begin{equation}\label{eq:def_joint_approx_full like}
\tilde{L}_n^{\text{RA}}(\theta) := L_{r(n)}^{(k(n))}(\theta) \uparrow L(\theta) \; \text{ as } n\to\infty.
\end{equation}
Applying the OSTE scheme to the sequence $\{\tilde{L}_n^{\text{RA}}(\theta)\}_{n\in\Z_+}$ yields a new estimator of the likelihood
\begin{equation}\label{eq:def_est_RA}
\tilde{L}_{\text{RA}}(\theta,N) := \tilde{L}_0^{\text{RA}}(\theta) + \frac{\tilde{L}_{N+1}^{\text{RA}}(\theta) - \tilde{L}_{N}^{\text{RA}}(\theta)}{p(N)}.
\end{equation}

Since the IE and RE estimators based on exact solutions are closely related to the ones based on approximate solutions, the rest of the paper will focus on the IA and RA estimators.

The relative efficiency of the IA and RA estimators is not obvious. On the one
hand, RA needs only one call to the matrix exponentiation algorithm for each
likelihood estimate, while IA needs $n_d$. Also, a likelihood estimate produced
by RA involves a smaller number of states, since from
\cref{eq:def_RE_merged_seq_SS} we see that
\[
|\statespace_{r}| \leq \sum_{i=1}^{n_d} |\statespace_{r}^i|.
\]
On the other hand, IA allows for finer tuning since each observation gets its
own joint sequence. For example,
transition probabilities for observations of the form $(x,x,\Delta t)$ in
models with low rates of transition can be well approximated by the probability
of not moving---i.e., $e^{q_{x,x}\Delta t}$---and therefore we can set both of
its offsets to a low value. In contrast, for RA we need to set its parameters
so that all transition probabilities are well approximated.

Due to these competing effects, the relative efficiency of IA and RA will
depend on the particulars of the problem at hand, as we shall see in
experiments later on. Nevertheless, a good rule of thumb is to use the RA
approach whenever
\[
|\statespace_0| \leq \frac{1}{3}\sum_{i=1}^{n_d} |\statespace_0^i|.
\]
This rule selects RA whenever there is sufficient redundancy in the base
state-spaces so that it is more efficient to work with the union of them.

\subsection{Doubly-monotone approximations}\label{subsec:monotone_approx_matrix_exp}

In this section we describe doubly-monotone approximations 
that can be used as the basis of
the IA and RA unbiased likelihood estimators described in the previous section.
In particular, we begin with a simple approach based on uniformization \citep{jensen53}, which is efficient for rate matrices that have low norm. We present two variants which differ in their applicability to more general classes of CTMCs, as well as in their fulfillment of the double monotonicity property. Additionally, we 
develop a novel matrix exponential approximation---the \emph{skeletoid}---which is doubly monotonic and excels in the high rate regime. Finally, we compare the efficiency of the two methods in practical
application to the exponentiation of rate matrices for a number of CTMCs.

\subsubsection{Uniformization}\label{subsubsec:uniformization}

The process of uniformization described in \cref{subsec:designing_sequences_SS} 
can be used to produce a monotone sequence of approximations to $e^{tQ}$.
Suppose $Q$ 
satisfies $\bar{q} = \inf_{x\in \statespace}q_{x,x}>-\infty$. Define $P:=I+Q/(-\bar{q})$, then
\[
M(t) = e^{tQ} = e^{\bar{q} t(I-P)} = \sum_{n=0}^\infty e^{\bar{q} t} \frac{(-\bar{q} t)^n}{n!} P^n,%
\]
and the elements of the sum on the right are all (entry-wise)
non-negative. Using this fact, we can obtain an increasing sequence of
approximate solutions,
\begin{equation}\label{eq:uniformization_approx_sol}
\forall s\in\Z_+: M^{(s)}(t) := \sum_{n=0}^s e^{\bar{q} t} \frac{(-\bar{q}t)^n}{n!} P^n,
\end{equation}
such that $M^{(s)}(t)\uparrow M(t)$ as $s\to\infty$. Moreover, for $r\in\Z_+$, let $P_r:=I+Q_r/(-\bar{q})$ (with $Q_r$ defined in \cref{prop:monotone_trunc}). Define
\begin{equation}\label{eq:uniformization_global_double_seq}
\forall r,s \in \Z_+, t\geq0: M_r^{(s)}(t) := \sum_{n=0}^s e^{\bar{q} t} \frac{(-\bar{q}t)^n}{n!} P_r^n.
\end{equation}
Then, for any $x,y\in\statespace$ and $t\geq 0$, the collection  $\{\mseq{r}{s}(t)_{x,y}\}_{r,s\in\Z_+}$ is doubly monotonic, as \cref{prop:doubly_monotone_global_unif} shows.

\begin{proposition}\label{prop:doubly_monotone_global_unif}
The collection in \cref{eq:uniformization_global_double_seq} is doubly monotone.
\end{proposition}

We now study the error incurred by uniformization. Let
$\|A\|_\infty$ be the $\ell_\infty$ operator norm of the matrix $A$,
\begin{equation}\label{eq:def_norm_ell_infty}
\|A\|_\infty := \sup_{x\in \statespace}\sum_{y\in\statespace} |A_{x,y}|.
\end{equation}
The $\ell_\infty$ error incurred by uniformization admits a simple bound.

\begin{proposition}\label{prop:unif_ell_infty_error}
Consider a possibly non-conservative rate matrix $Q$ with associated state space $\statespace$. Suppose there exist $\bar{q}\in(-\infty,0)$ such that $\bar{q} \leq \inf_{x\in \statespace} q_{x,x}$. Let $\lambda:=(-\bar{q})t$ and $F(n;\lambda)$ be the cumulative distribution function of a $\mathrm{Poisson}(\lambda)$ distribution. Then,
\begin{equation}\label{eq:unif_ell_infty_error}
\|M(t) - M^{(s)}(t)\|_\infty \leq 1-F(s;\lambda).
\end{equation}
\end{proposition}
Given any $\epsilon>0$, we can invert \cref{eq:unif_ell_infty_error} to find $s\in\Z_+$ such that the error is less than $\epsilon$. Indeed, let
\[
F^{-1}(p;\lambda) := \inf\{s\in\Z_+: p\leq F(s;\lambda)\}
\]
be the quantile function of a $\text{Poisson}(\lambda)$
distribution. Then, setting
\begin{equation}\label{eq:def_seq_approx_sol_poisson_quantile}
s = F^{-1}(1-\epsilon;\lambda),
\end{equation}
guarantees that the $\ell_\infty$ error is bounded by $\epsilon$. In particular, note that by imposing $\epsilon=\epsilon(k)=e^{-\kappa k}$, for any $\kappa>0$, the second requirement in \cref{prop:doublymonotonevariance} is satisfied.

For the purposes of pseudo-marginal inference for CTMCs (\cref{subsec:doubly_monotone}), the uniformization approach presented so far---which we call \emph{global uniformization}---assumes that the rate matrix $Q$ for the infinite space $\statespace$ satisfies $\inf_{x\in \statespace}q_{x,x}>-\infty$. Although appealing due to its double monotonicity, the global uniformization assumption does not hold in many infinite state space CTMCs of interest---like the ones included in our experiments (\cref{sec:experiments}), and in particular in the running example of the SSIR model. 

\emph{Sequential uniformization} is an alternative approach that is applicable to CTMCs that are not uniformizable; i.e., models for which $\inf_{x\in \statespace}q_{x,x}=-\infty$. Here, a decreasing sequence of lower bounds $\bar{q}_r := \inf_{x\in \statespace_r}q_{x,x}$ is computed for every truncated rate matrix $Q_r$. The corresponding double sequence of approximations becomes
\begin{equation}\label{eq:uniformization_seq_double_seq}
\forall r,s \in \Z_+: M_r^{(s)}(t) := \sum_{n=0}^s e^{\bar{q}_r t} \frac{(-\bar{q}_rt)^n}{n!} \tilde P_r^n,
\end{equation}
where $\tilde P_r:=I+Q_r/(-\bar{q_r})$.
For every $r\in\Z_+$, $t\geq0$, and $x,y\in\statespace_r$, $\{\mseq{r}{s}(t)_{x,y}\}_{s\in\Z_+}$ is increasing. However, the sequential approach is not doubly monotonic. Consider the following  counterexample. Let $(x,x,1)$ be the observation of interest, so that the base state space is a singleton $\statespace_0=\{x\}$ (see \cref{subsec:designing_sequences_SS}). Suppose further that the associated truncated rate matrix is the one-by-one matrix $Q_0=(-1)$.  This yields $(\mseq{0}{0}(1))_{1,1} = e^{-1}$.  Now suppose that $\statespace_1$ has an additional state, and that
\[
Q_1 = 
\begin{pmatrix}
-1 & 0 \\
0 & -10
\end{pmatrix}.
\]
Then, the base sequential uniformization approximation gives
\[
\mseq{1}{0} = 
\begin{pmatrix}
e^{-10} & 0 \\
0 & e^{-10}
\end{pmatrix}.
\]
Thus, $(\mseq{1}{0}(1))_{1,1}<(\mseq{0}{0}(1))_{1,1}$, so sequential uniformization is not doubly monotonic. 

In practice, situations that break double monotonicity for the sequential approach are rare. And, for observations $(x_{i-1},x_i,\Delta t_i)$ that do exhibit this phenomenon, we can always set the values of the sequence $k_i(n)$ high enough so that---up to numerical accuracy---we effectively fall back to using the IE estimator (which does not require double monotonicity). As we show in \cref{sec:experiments}, this implementation of sequential uniformization works well for models with low $\bar{q}_r$ values. As a consequence, we take uniformization to mean sequential uniformization unless mentioned otherwise.

Besides the issue with double monotonicity, an important drawback of uniformization is that it is inefficient when
$\lambda=-\bar{q}t$ is large. 
In particular, for fixed $\epsilon$, the $\mathrm{Poisson}(\lambda)$ quantile function
behaves like $F^{-1}(1-\epsilon,\lambda) = O(\lambda)$ 
\citep{giles2016} and so we must compute $s=O(\lambda)$ terms in
\cref{eq:uniformization_approx_sol} if we wish to have $\ell_\infty$ error below $\epsilon$. Therefore uniformization
can be prohibitively slow in cases with large $\lambda$, such as the
Schl{\"o}gl model (described in \cref{sec:experiments}). 

Finally, we note that a naive implementation of
 \cref{eq:uniformization_approx_sol} is prone to
numerical instability at high values of $\lambda$. In our experiments we use the
more robust algorithm described in \cite{sherlock2018direct}.

\subsubsection{The skeletoid}\label{subsec:Skeletoid}

In the following we describe a novel monotone
approximation to the matrix exponential---which we 
refer to as the \textit{skeletoid}---that 
is doubly monotonic and requires only $O(\log\lambda)$ computation, therefore excelling in the large $\lambda$ setting. 

For any $\delta>0$, let $S(\delta)$ be a matrix with entries $x,y\in\statespace$ defined by
\begin{equation}\label{eq:def_S_delta}
(S(\delta))_{x,y} := 
\bracearraycond{e^{q_{x,x} \delta} & \text{if }x = y, \\
q_{x,y}\delta e^{q_{x,x}\delta} & \text{if }q_{x,x}=q_{y,y}, \\
	q_{x,y} \frac{ e^{q_{y,y} \delta} - e^{q_{x,x} \delta}}{q_{y,y} - q_{x,x}} &\text{otherwise.}} 
\end{equation}
The \emph{skeletoid approximation} to the transition function is given by
\begin{equation}\label{eq:Skeletoid_approx_sol}
M^{(s)}(t) := S(t2^{-s})^{2^s}, \hspace{1em} s\in\Z_+.
\end{equation}
Note that, once $S(t2^{-s})$ is computed, its $2^s$-th power can be evaluated using only $s$ matrix-matrix multiplications via repeated squaring. Indeed, for $l=0,1,\dots,s-1$, set
\[
S(t2^{-s})^{2^{l+1}} = \left[ S(t2^{-s})^{2^l} \right]^2.
\]
\cref{prop:Skeletoid_monotonicity} shows that $S(t2^{-s})^{2^s}$ defines a novel 
monotone increasing sequence of
approximations to the transition function of a CTMC.

\begin{proposition}\label{prop:Skeletoid_monotonicity}
If $Q$ is non-explosive, then for all $x,y \in \statespace$ and $t\geq 0$,
\begin{equation}\label{eq:Skeletoid_monotonicity}
\left( S(t2^{-s})^{2^s}\right)_{x,y} \uparrow (M(t))_{x,y},\;\text{ as }s\to \infty.
\end{equation}
\end{proposition}

The proof for \cref{prop:Skeletoid_monotonicity} can be found in \cref{app:proofs}. 
It is important to 
highlight that the non-explosivity condition of \cref{prop:Skeletoid_monotonicity} 
is qualitatively less stringent than the assumption that $\inf_x q_{x,x}>-\infty$, which is required by the global uniformization approach.

Additionally, for $r\in\Z_+$, let $S_r(\delta)$ be the matrix constructed from \cref{eq:def_S_delta} using the truncated rate matrix $Q_r$. Let
\begin{equation}\label{eq:Skeletoid_double_seq}
\forall r,s \in \Z_+, t\geq0: M_r^{(s)}(t) := S_r(t2^{-s})^{2^s}
\end{equation}
be the collection of approximations induced by the skeletoid method. \cref{prop:doubly_monotone_skeletoid} shows that this collection is doubly monotonic.

\begin{proposition}\label{prop:doubly_monotone_skeletoid}
The skeletoid double sequence in \cref{eq:Skeletoid_double_seq} is doubly monotone.
\end{proposition}

In order to find an $s\in\Z_+$ that achieves an $\ell_\infty$ error lower than any given $\epsilon>0$, we require an error bound that admits a straightforward
inverse. The following proposition gives such a bound.

\begin{proposition}\label{prop:Skeletoid_err_bound}
If $\bar{q}:=\inf_x q_{x,x}>-\infty$, then for all $s\in\Z_+$
\begin{equation}\label{eq:Skeletoid_error_upper_bound}
\left\|M(t) - S(t2^{-s})^{2^s}\right\|_\infty \leq (\bar{q} t)^22^{-(s+1)} + O((\bar{q} t)^32^{-2s}).
\end{equation}
\end{proposition}

Discarding the higher-order terms in \cref{eq:Skeletoid_error_upper_bound} and solving for $s$ yields the rule
\begin{equation}\label{eq:def_seq_approx_sol_Skeletoid}
s = \left\lceil\log_2\left(\frac{(\bar{q}t)^2}{2\epsilon}\right)\right\rceil.
\end{equation}
\cref{eq:def_seq_approx_sol_Skeletoid} shows that the
cost of achieving a given error is $O(\log\lambda)$ 
(with $\lambda:=(-\bar{q})t$ as before), which explains the differences in
performance between the skeletoid and uniformization methods in the large $\lambda$ setting. Also, and notwithstanding the fact that \cref{eq:Skeletoid_error_upper_bound} already implies exponential decay of the error as $2^{-s}$, we can achieve any other rate of decay $\kappa>0$ by setting $\epsilon=\epsilon(k)=e^{-\kappa k}$ in \cref{eq:def_seq_approx_sol_Skeletoid} (c.f. \cref{prop:doublymonotonevariance}).

Finally, we note that the skeletoid approximation is related to the ``scaling and squaring'' 
method, a common trick used as part of various algorithms for computing the
matrix exponential \citep{moler_nineteen_2003}. Scaling and squaring reduces the problem of
computing $e^{tQ}$ to that of obtaining the transition function $e^{\delta Q}$
of the $\delta$-skeleton---with $\delta:=t2^{-s}$---by applying the identity 
$e^{tQ} = (e^{\delta Q})^{2^s}$ (a consequence of Chapman-Kolmogorov). The key difference here is that $S(\delta)$ is a specific computationally cheap 
approximation to the $\delta$-skeleton solution $e^{\delta Q}$---hence  the
name ``skeletoid.'' This specific choice is critical to establishing double monotonicity in a general non-explosive setting,  making the skeletoid approximation particularly suitable to pseudo-marginal methods.

\subsubsection{Performance comparisons}

\begin{figure}[t]
\centering
\includegraphics[width=\textwidth]{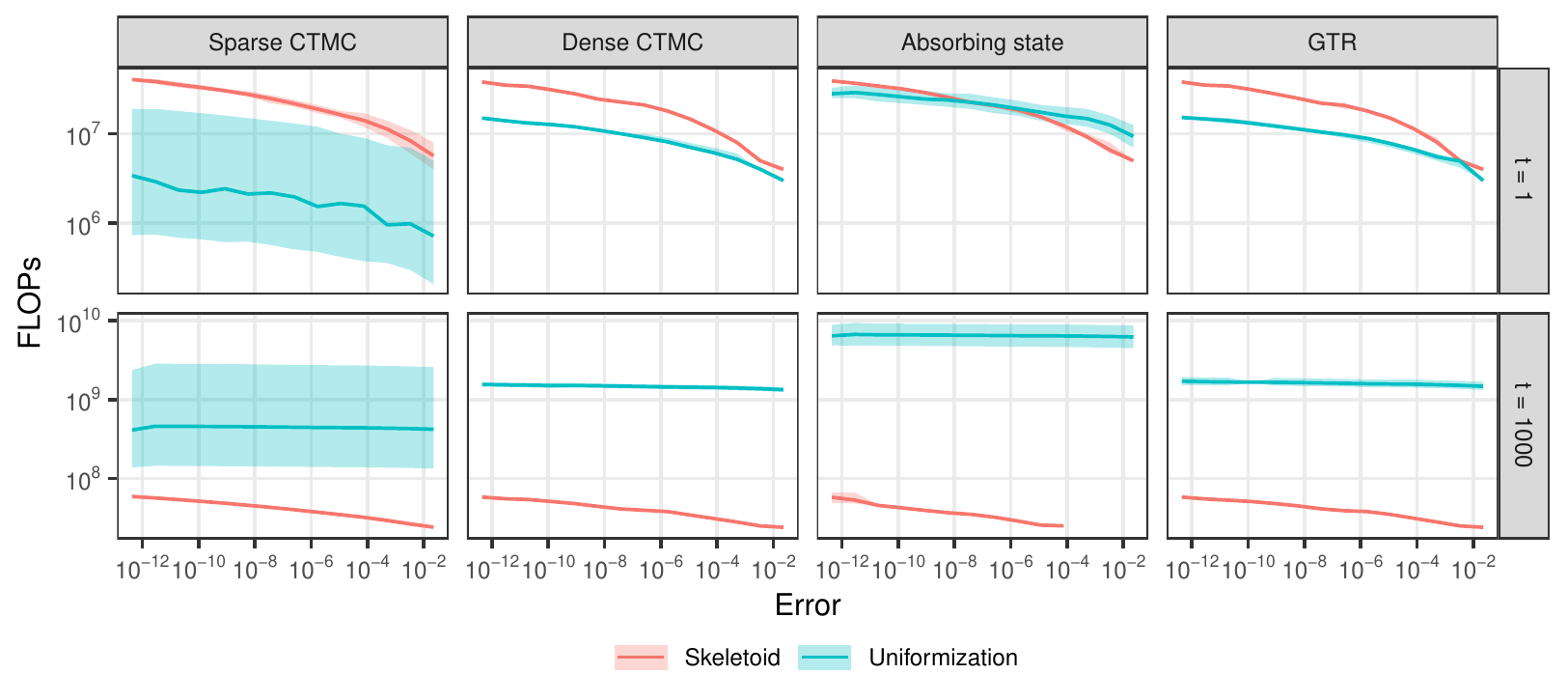}
\caption{FLOPs associated to matrix multiplications needed to compute approximate solutions to the transition function for random rate matrices of size 100, versus the $\ell_\infty$ error incurred. Rows represent increasing time exponents ($t$ in $e^{tQ}$). Lines show mean error over 100 repetitions of a given combination, while bands give a 95\% coverage range. The mean rate of each random rate matrix is normalized to 1.}
\label{fig:solvers_speed_vs_acc_expm_100}
\end{figure}

\cref{fig:solvers_speed_vs_acc_expm_100} shows the performance of
uniformization and skeletoid in approximating the exponential of random
rate matrices $Q$ drawn at random from four qualitatively different classes of CTMCs. In
the sparse class, for each row we sample a random integer giving the number of non-zero
entries, and then sample the positions of these elements in each row at
random. In contrast, matrices in the dense clase have only non-zero entries. 
In the absorbing state class we make the first state absorbing, meaning that
$q_{1j}=0$ for all $j$; then, we draw a random dense rate sub-matrix for the
remaining states, and finally complete the first column to satisfy the
conservative property. Lastly, in the General Time Reversible (GTR) class we
simulate matrices satisfying detailed balance: $\mu_x q_{x,y}=\mu_yq_{y,x}$ for
all $x,y$, for some probability vector $\mu$. In all these
classes, the non-zero off-diagonal elements of the matrices are drawn from
$\mathrm{Exp}(1)$ distributions; then, the diagonal is set to enforce the conservative
property, and finally the matrix is scaled to achieve a mean rate of one. 
Note that the skeletoid is much faster than uniformization in
the $t=1000$ setting, across all four CTMC classes, whereas in the $t=1$ setting, the ranking of the two methods depend on the computational budget.

\begin{figure}[t]
\centering
\includegraphics[width=\textwidth]{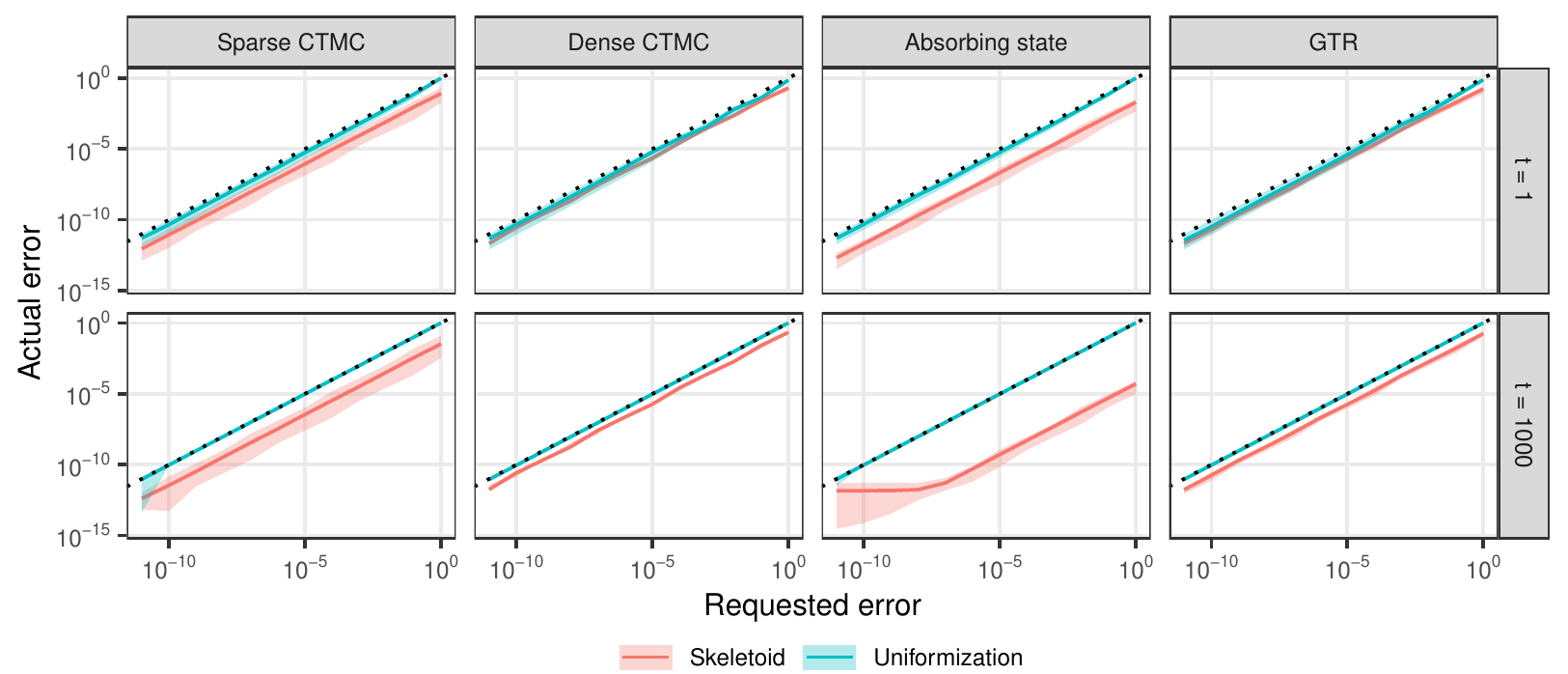}
\caption{Realized $\ell_\infty$ error versus requested error for the two algorithms when applied to random rate matrices of size 100. Rows represent increasing time exponents ($t$ in $e^{tQ}$). Lines show mean error over 100 repetitions of a given combination, while bands denote the range of the error. The dotted diagonal line represents ideal behavior. The mean rate of each rate matrix is normalized to 1.}
\label{fig:check_auto_tune}
\end{figure}

\cref{fig:check_auto_tune} shows the result of using \cref{eq:def_seq_approx_sol_poisson_quantile,eq:def_seq_approx_sol_Skeletoid} to produce approximations with a given error bound. We see a high
level of agreement between requested and realized errors in most of the
settings, with a bias towards more conservative outcomes when disagreement
occurs.

Notwithstanding the usefulness of \cref{eq:unif_ell_infty_error,eq:Skeletoid_error_upper_bound}
to provide bounds on errors that can be evaluated before doing any 
computation, monotone increasing approximations (not necessarily doubly monotone) to the matrix exponential of rate
matrices---like skeletoid and uniformization---offer a simple way of obtaining 
the exact error of a given approximation, since
\[
\|M(t) - M^{(s)}(t)\|_\infty = \sup_{x\in \statespace}\sum_{y\in\statespace}[M(t)_{x,y} - M^{(s)}(t)_{x,y}] \leq 1 - \inf_{x\in \statespace}\sum_{y\in\statespace}M^{(s)}(t)_{x,y},
\]
where the inequality arises from the possibility that $Q$ is non-conservative. The right-hand side is computable since the formula does not depend on $M(t)$. The inequality becomes an equality in the case of a conservative $Q$, and in this case a stronger statement holds for all $x \in \statespace$:
\[
\sum_{y\in\statespace}[M(t)_{x,y} - M^{(s)}(t)_{x,y}] = 1 - \sum_{y\in\statespace}M^{(s)}(t)_{x,y}.
\]

Finally, we note that it is always possible to improve a collection of monotone
increasing approximations by computing all of them in parallel and returning their
element-wise maximum, which will again be monotone increasing. However, in our
tests the overhead needed to parallelize the skeletoid and uniformization methods was not
compensated by the efficiency gains obtained. However, it is possible that for models not considered in this study, 
an algorithm returning the element-wise maximum of the uniformization, skeletoid,
and potentially other monotone approximations could be advantageous.

\section{Related work}

\cite{georgoulas_unbiased_2017} were the first to describe a pseudo-marginal
method to perform exact inference for CTMCs on countably infinite spaces by
debiasing a sequence of likelihood estimates arising from an increasing
sequence of truncated state-spaces. They used the randomly-truncated series
debiasing scheme first introduced by \cite{mcleish2011}, which has the downside
of requiring an unbounded random number of calls to the matrix
exponentiation algorithm. In contrast, the OSTE scheme needs at most $3$ calls to
these routines. Additionally,
\cite{georgoulas_unbiased_2017} used a stopping time whose distribution has
super-exponential tails. 
This introduces the risk of obtaining likelihood estimates with infinite
variance, unless one is able to prove that the approximating sequence converges
faster. As discussed in \cref{subsec:designing_sequences_SS}, there is
evidence that the sequences converge at super-exponential speed, but the
precise rate must still be estimated on a case-by-case basis.

Building on \citet{georgoulas_unbiased_2017}, \cite{sherlock_exact_2019} introduced the MESA
algorithm. This method expands the state space of 
the sampler to include the random index $r$ of the smallest truncated state
space that contains the full (unobserved) path of the process for all
observation pairs. This strategy is appealing because it dispenses with the
need to design debiasing methods with appropriate stopping times. The variant
nMESA is also proposed, which has a different index $r_i$ for each observation
pair. In terms of computational cost, the sampler used in MESA requires $3$
calls to matrix exponentiation routines, which are based on uniformization and
scaling-and-squaring.
On the other hand, because they are not pseudo-marginal samplers, neither MESA
nor nMESA can take advantage of the accuracy relaxation technique described in
\cref{subsec:monotone_approx_matrix_exp}. Thus, they require computing matrix
exponentials up to numerical accuracy.

Previous work has explored the use of particle MCMC to design positive unbiased
estimators in the context of infinite state space CTMC.
The most straightforward way to use particle MCMC for infinite CTMC is to
impute full paths and set weights to zero when a particle does not match an
observed end-point \citep{Saeedi2011Priors,golightly_bayesian_2011}. However
this approach can suffer from severe weight degeneracy, including particle
population collapse where all weights are equal to zero. A more sophisticated
approach builds a martingale to guide particle paths toward the end-point
\citep{Hajiaghayi2014Efficient}, but this requires the construction of a
problem-specific potential function.

The literature on exact inference for CTMCs on finite state spaces is extensive
\citep[see e.g.][]{geweke1986exact,fearnhead2006exact,ferreira2008bayesian},
but methods therein cannot be extended to the countably infinite case in a
straightforward manner. In particular, as we have previously mentioned,
uniformization cannot be applied to infinite state-space CTMCs whose rate
matrices do not admit a uniform bound on the diagonal. Another line of
work has built advanced extensions for the related but distinct problem of
end-point simulation \citep{rao_fast_nodate,zhang_efficient_2020}.

\section{Experiments}\label{sec:experiments}

\begin{algorithm}[t]
\DontPrintSemicolon
\SetKwInOut{Input}{input}
\SetKwInOut{Output}{output}
\Input{$n_\text{samples},\theta_0,\Sigma_{\text{prop}},p,\{r(n),k(n)\}_{n\in\Z_+}$}
$\mathcal{C}\gets \{\statespace_0\}$\Comment*{construct and store $\statespace_0$}
$\bar{r}\gets 0$\Comment*{index of the largest state-space available}
$N\sim \text{Geom}(p)$\;
$L_\text{cur}\gets \tilde{L}_{\text{RA}}(\theta_0,N)$\Comment*{using \cref{eq:def_est_RA} with $\{r(n),k(n)\}_{n\in\Z_+}$}
\For{$s=1,2,\dots,n_\text{samples}-1$}{
    $\theta_\text{prop}|\theta_{s-1}\sim\mathcal{N}(\theta,\Sigma_\text{prop})$\;
    $N\sim \text{Geom}(p)$\;
	\If(\Comment*[f]{need to build larger state-spaces}){$r(N+1) > \bar{r}$}{
        $\mathcal{C}\gets \mathcal{C} \cup \{\statespace_r\}$ for all $r\in\{\bar{r}+1,\dots, r(N+1)\}$\;
        $\bar{r}\gets r(N+1)$\;
    }
    $L_\text{prop}\gets \tilde{L}_{\text{RA}}(\theta_\text{prop},N)$\;
    $A \gets1\wedge \frac{\pi(\theta_\text{prop})L_\text{prop}}{\pi(\theta_{s-1})L_\text{cur}}$\;
	$U \sim U[0,1]$\;
    \uIf(\Comment*[f]{proposal accepted}){$U\leq A$}{
        $\theta_s\gets \theta_\text{prop}$\;
        $L_\text{cur}\gets L_\text{prop}$\;
    }
    \Else(\Comment*[f]{proposal rejected}){
        $\theta_s\gets \theta_{s-1}$\;
    }
}
\Output{$\{\theta_s\}_{s=0}^{n_\text{samples}-1}$.}
\caption{Metropolis pseudo-marginal sampler with RA estimator}
\label{algo:overview_sampling} 
\end{algorithm}

We evaluate our methods using experimental benchmarks introduced in
\cite{georgoulas_unbiased_2017}\footnote{Code and data available at \url{https://github.com/ageorgou/roulette/}} and \cite{sherlock_exact_2019}.\footnote{The authors kindly shared their code and data with us.} We run a random
walk Metropolis sampler targeting the augmented posterior distribution
$\pi(\theta,U|\D)$ of \cref{eq:pseudo_marginal_aug_target_density}. The
proposal for $\theta$ uses a multivariate normal with covariance matrix
$\Sigma_{\text{prop}}=\alpha \var(\theta|\D)$, for some $\alpha>0$ and with
$\var(\theta|\D)$ estimated in a trial run. For the auxiliary variables $U$ we
take independent draws from the distribution used by the OSTE scheme, which we
set to $\text{Geom}(p)$ distributions with $p\in(0,1)$.
\cref{app:tuning_stopping_times} contains detailed explanations on how to tune
all the parameters involved. An overview of the process to run the Metropolis
sampler using the RA estimator is shown in \cref{algo:overview_sampling}. The
process used for a sampler using the IA estimator is similar but involves also
looping over observations.

In order to fairly compare the efficiency of algorithms implemented in
different languages, we use the effective sample size per billion floating
point operations (ESS/GFLOPs), when considering FLOPs used in matrix
multiplications within matrix exponentiation algorithms. The reason to focus on
these operations is that they account for most of the computational cost
incurred by the samplers. The ESS is measured using the method implemented in the
\verb|R| package \verb|mcmcse| \citep{flegal2020mcmcse} and described by
\cite{gong2016practical}.
Due to the importance of correctly tuning samplers to obtain good results, we only compare against other methods that have been previously tuned for and tested in a given dataset by their respective authors.

\begin{figure}[t]
\centering
\includegraphics[width=0.7\textwidth]{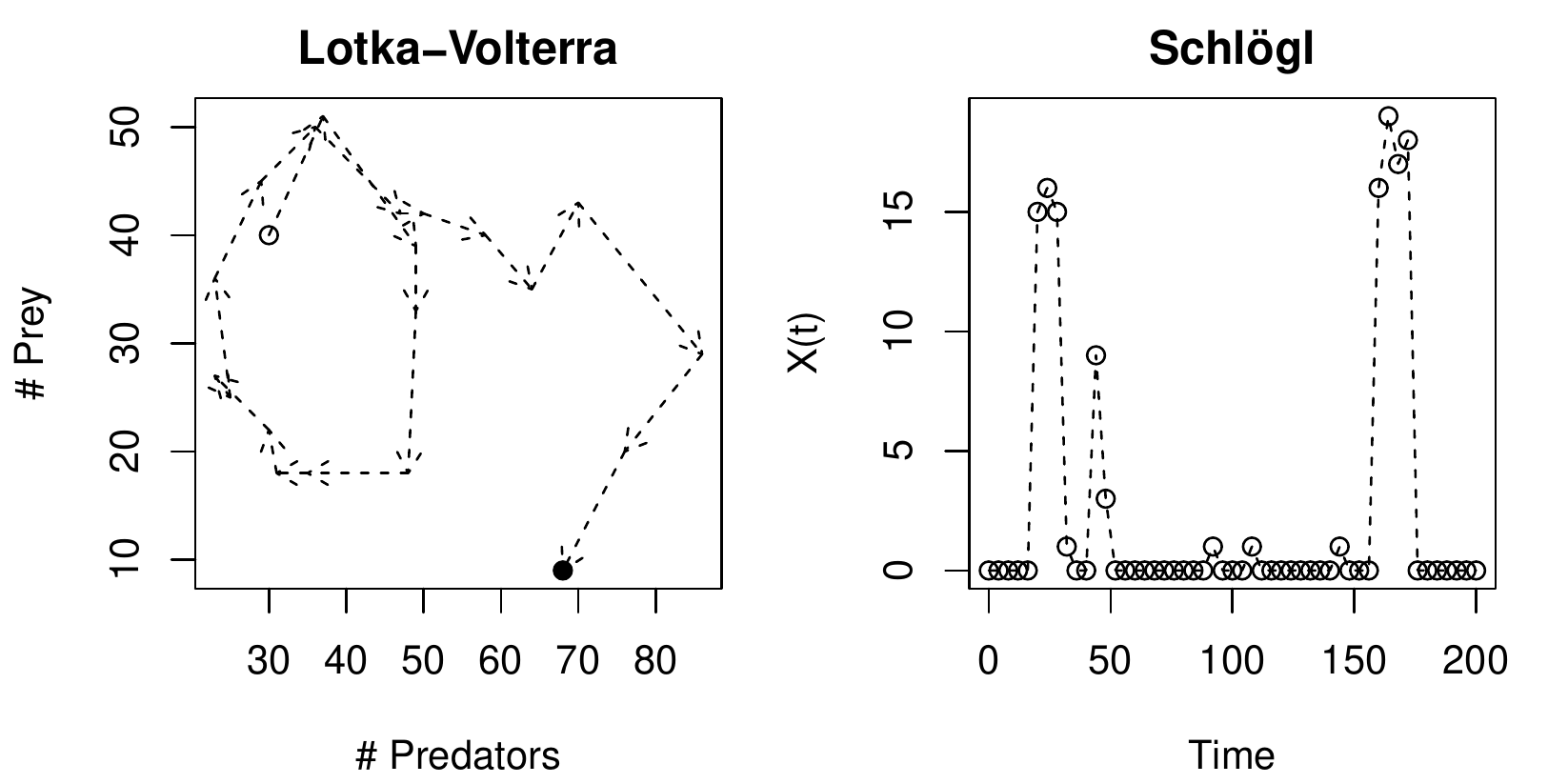}
\includegraphics[width=0.7\textwidth]{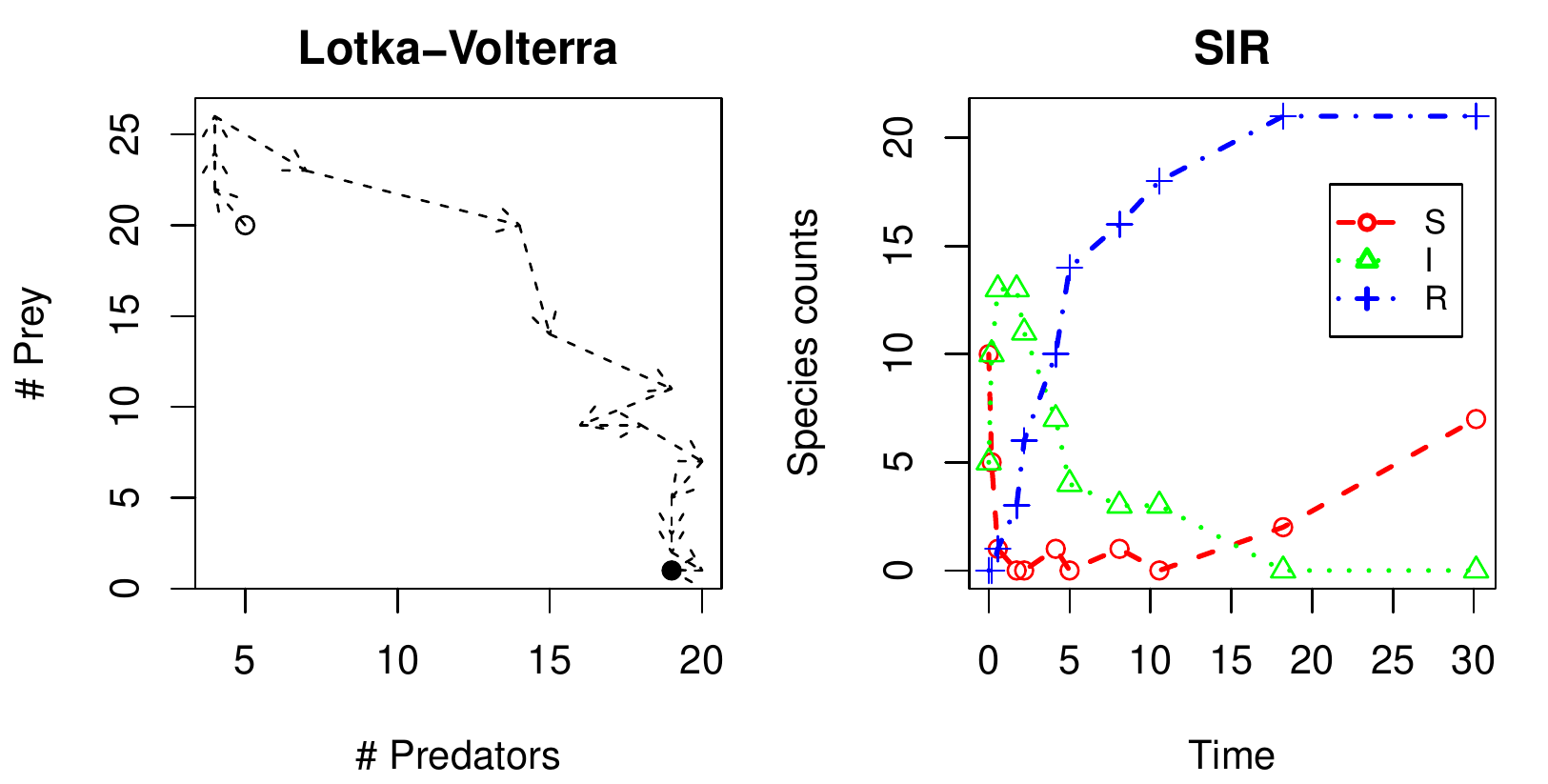}
\caption{Plots of the datasets considered in the experiments. Dashed lines are used to emphasize that the path of the process is unknown between observations. \textbf{Top:} datasets from \citet{sherlock_exact_2019} for the Lotka--Volterra (3 reactions) and Schl{\"o}gl models. \textbf{Bottom:} datasets from \citet{georgoulas_unbiased_2017} for the Lotka--Volterra (4 reactions) and SIR models}
\label{fig:datasets}
\end{figure}

\subsection{Schl{\"o}gl model}

The Schl{\"o}gl model is defined by four reactions
\begin{equation}\label{react:schloegl}
\begin{array}{ccc}
A + 2X &\xrightarrow{\mathmakebox[7em]{\theta_1 X(X-1)/2}}& 3X \\
3X &\xrightarrow{\theta_2 X(X-1)(X-2)/6} & 2X + A \\
B &\xrightarrow{\mathmakebox[7em]{\theta_3}}& X \\
X &\xrightarrow{\mathmakebox[7em]{\theta_4X}}& B \\
\end{array}
\end{equation}
where the number of $A$ and $B$ molecules are kept constant, so that the only
species evolving in time is $X$. The model is known in the mathematical biology
literature as the canonical example of a chemical reaction system exhibiting
``bi-stability'' \citep{vellela2009schloegl}. This means that for some
parameter configurations, the system oscillates between two meta-states or
regimes, one with very low or zero counts of $X$, and one with high counts.

For the purpose of inference, it is simpler to work with the birth--death
version of the process, obtained by collapsing the reactions that increase $X$
and the ones that decrease it
\begin{align}
\lambda_{x} &= \theta_1\I\{x\geq2\} x(x-1)/2 + \theta_3 \\
\mu_{x} &= \theta_2 \I\{x\geq3\}x(x-1)(x-2)/6 + \I\{x\geq1\}\theta_4
\end{align}
for $x\in\Z_+$.

\begin{figure}[t]
\centering
\includegraphics[width=0.9\textwidth]{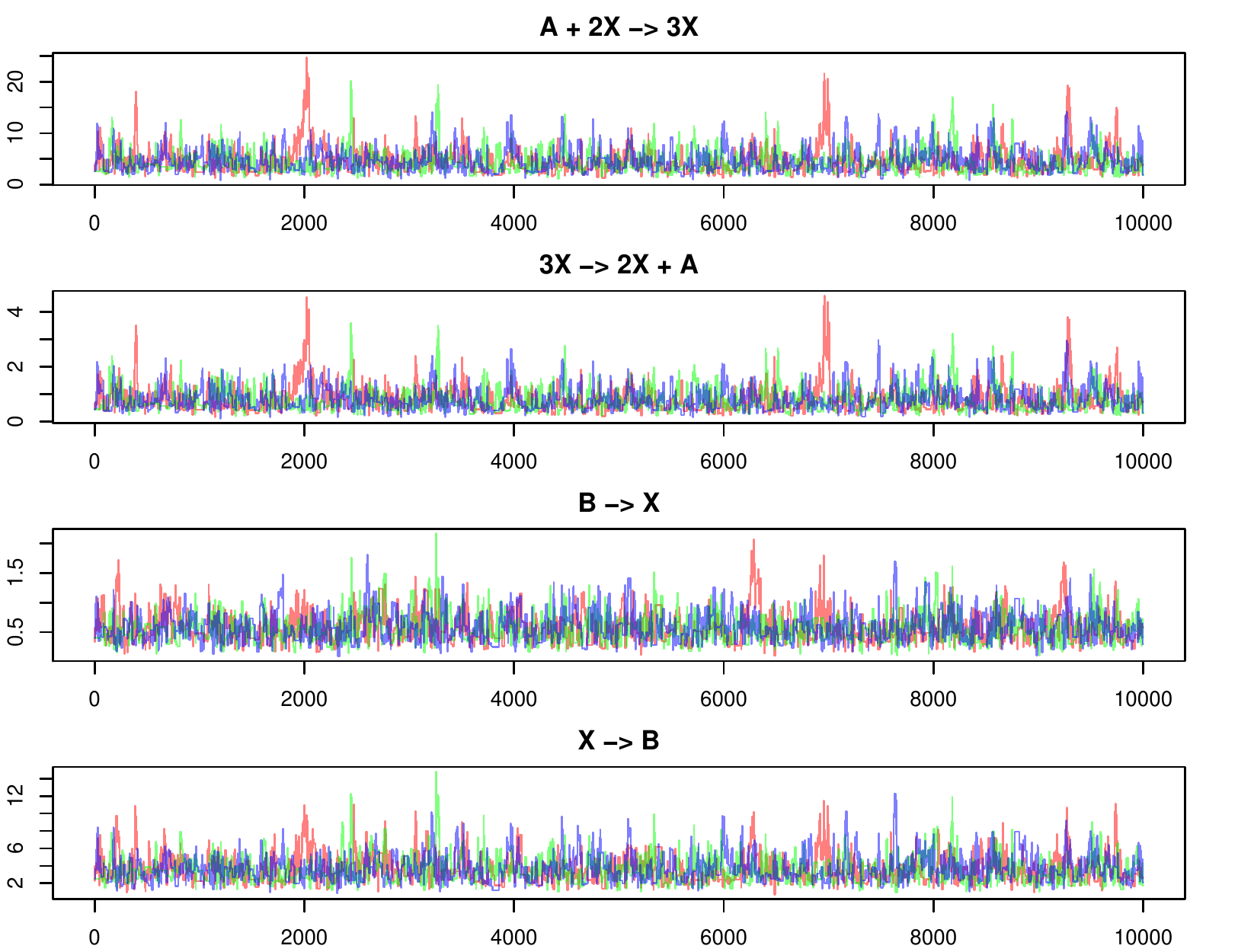}
\caption{Trace plots for 10,000 samples taken using 3 parallel chains of the RA pseudo-marginal sampler with Skeletoid targeting the posterior distribution of the Schl{\"o}gl model.}
\label{fig:mcmc_traceplots}
\end{figure}

\begin{figure}[t]
\centering
\includegraphics[width=0.8\textwidth]{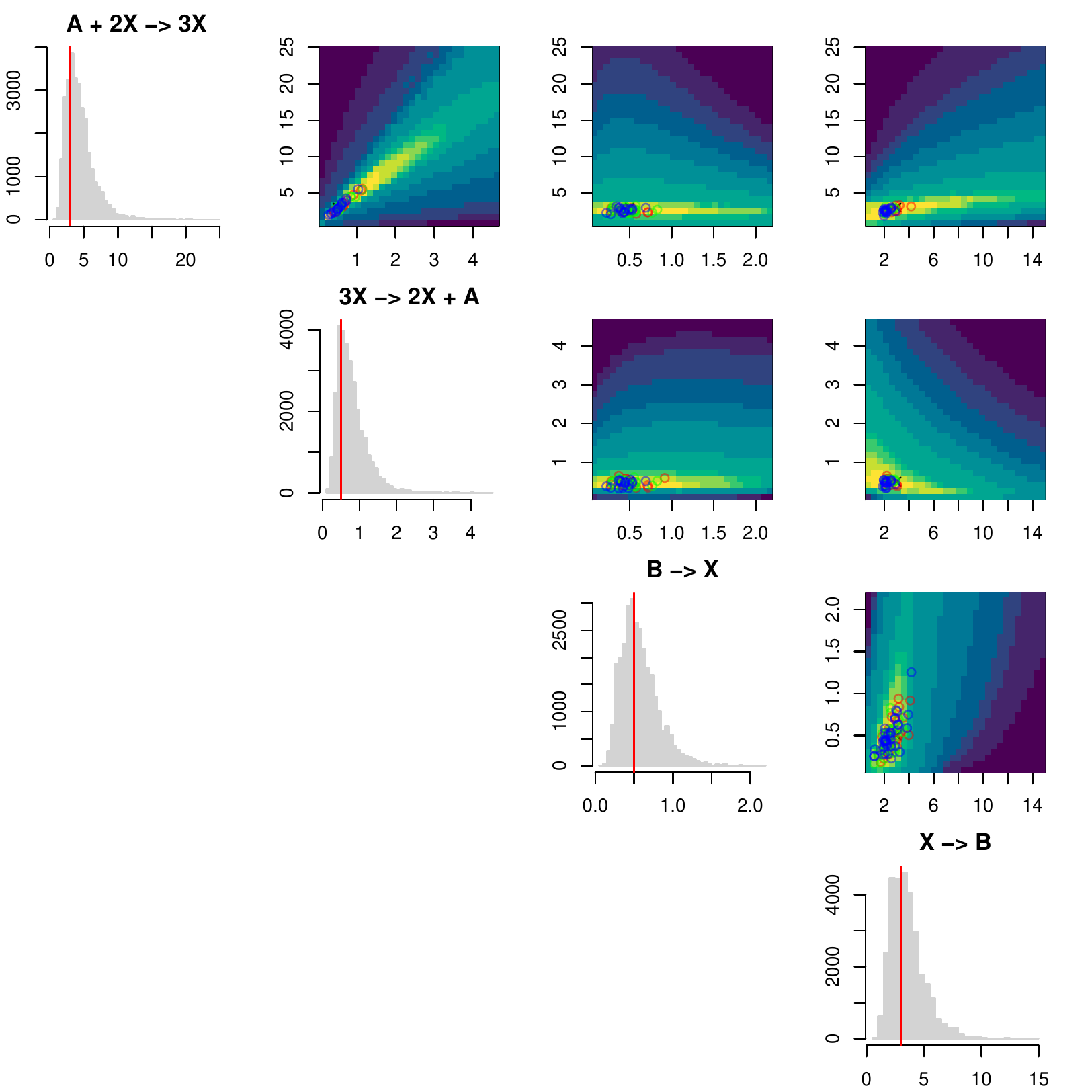}
\caption{Histograms of the samples obtained from the experiment described in \cref{fig:mcmc_traceplots}. Heatmaps show bi-variate views of the approximated posterior when using the approximate likelihood $L_r^k$ (\cref{eq:def_joint_approx_like}) with high values of $r,k$. Vertical lines in histograms and X's in bi-variate plots show location of data-generating parameter. Points in bi-variate plots show selected samples from every chain.}
\label{fig:density_plots}
\end{figure}

To give a more detailed description of the performance of our pseudo-marginal
approach, we first focus on the results obtained for the Schl{\"o}gl model,
then present summaries for several models. The dataset we use is taken from
\citet{sherlock_exact_2019}, and is depicted in \cref{fig:datasets}. Note that
it exhibits the bi-stability phenomenon described above. These data are taken
every $\Delta t=4$ from a path simulated with 
$\theta_\text{true}=(3, 0.5, 0.5, 3)$. \citet{sherlock_exact_2019} use a Metropolis 
sampler on the log-transformed variables $\xi=\log(\theta)$, and place independent
standard normal priors on $\xi$. Since we work directly with $\theta$, we use 
i.i.d.\ standard log-normals. After
following the tuning procedure outlined in \cref{app:tuning_stopping_times}, we
run three parallel chains for 10,000 iterations each, initialized at
$\theta_\text{true}$. \cref{fig:mcmc_traceplots} shows the trace-plots for an
RA sampler using the skeletoid method, evidencing good mixing in all
parameters and chains. \cref{fig:density_plots} summarizes the approximated
posterior. Note how there is almost perfect correlation between $\theta_1$ and
$\theta_2$, demonstrating that having a proposal with non-diagonal covariance
is useful.

\begin{figure}[t]
\centering
\includegraphics[width=\textwidth]{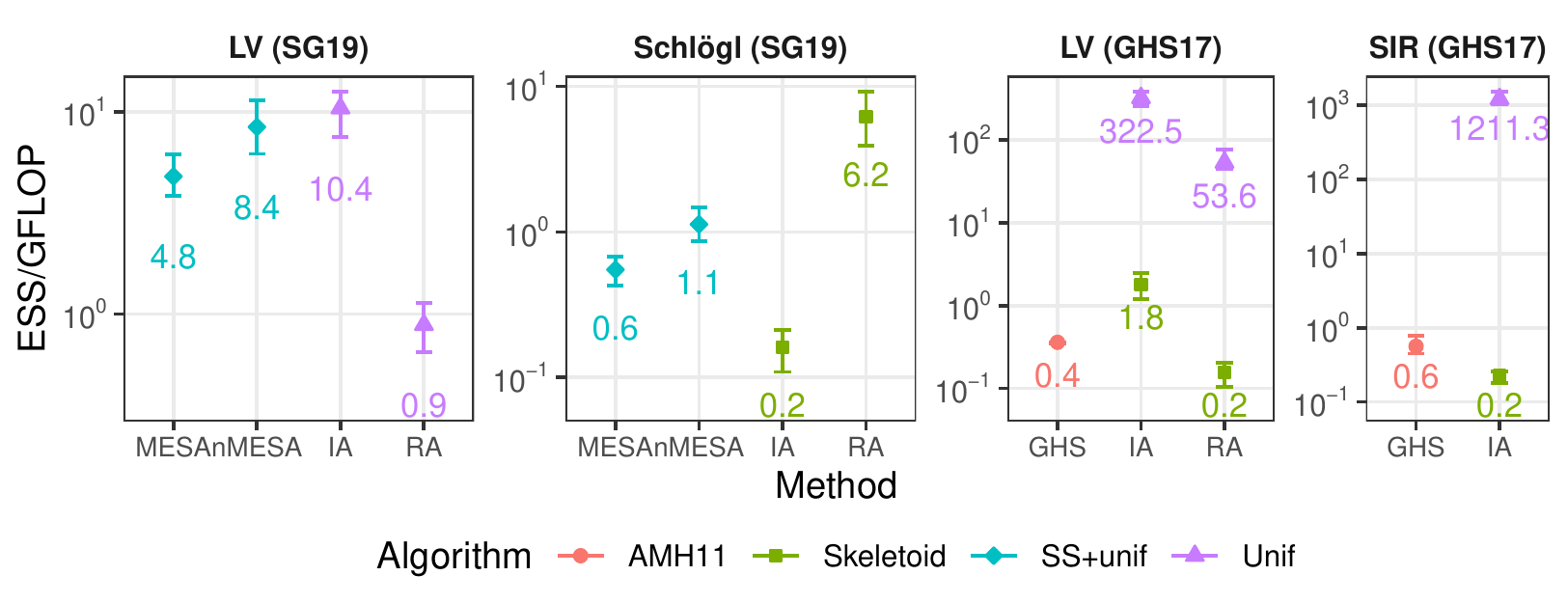}
\caption{Comparison of ESS/GFLOP between our proposed pseudo-marginal samplers and competing methods. The intervals denote the range across 30 repetitions, while points and text correspond to the means.}
\label{fig:benchmarks}
\end{figure}

The second panel in \cref{fig:benchmarks} shows the ESS/GFLOP achieved by the IA 
and RA samplers with the skeletoid, compared to the results obtained for the MESA 
and nMESA methods. Our RA sampler achieves an improvement of more than 450\% over 
nMESA.

\subsection{Lotka--Volterra model}

The Lotka--Volterra (LV) LV process has state-space $\statespace=\Z_+^2$; i.e., defined by boundaries $B^l=(0,0)$ and $B^u=(\infty,\infty)$. We consider two versions of this model with different numbers of reactions.

\subsubsection{Three reactions}

\citet{sherlock_exact_2019} use a version of LV containing three distinct reactions
\begin{equation}\label{react:LV3}
\begin{array}{cccl}
R &\xrightarrow{\mathmakebox[2em]{\theta_1 R}}& \emptyset &  \text{(death of predator)} \\
R + Y &\xrightarrow{\theta_2 RY} & 2R & \text{(predator eats prey \& reproduces)} \\
Y & \xrightarrow{\mathmakebox[2em]{\theta_3Y}} & 2Y & \text{(birth of prey)}
\end{array}
\end{equation}
Here, $R$ denotes the number of predators and $Y$ the size of the prey
population. For our experiments, we focus on the LV20 dataset---shown in the top-left panel
of \cref{fig:datasets}---containing 20
observations obtained at regular time intervals $\Delta t=1$ of an LV path
simulated with $\theta_\text{true}=(0.3, 0.4, 0.01)$. Again, \citet{sherlock_exact_2019} use a Metropolis 
sampler on the log-transformed variables $\xi$, putting independent Gaussian
priors with unit variance and means $(\log(0.2),\log(0.2),\log(0.02))$. To match these, we put independent log-normal priors on $\theta$
with the same parameters.

In the first panel of \cref{fig:benchmarks} we show that the IA sampler with 
uniformization achieves more than a 20\% average increase in ESS/GFLOP when 
compared agains nMESA, although the ranges do not show a strong separation. 
These results also show that the relative performance between IA and RA samplers 
depends on the particulars of the model and data.

\subsubsection{Four reactions}

\citet{georgoulas_unbiased_2017} use a version of the LV model with four reactions
\begin{equation}\label{react:LV4}
\begin{array}{cccl}
R &\xrightarrow{\mathmakebox[2em]{\theta_1 R}}& \emptyset &  \text{(death of predator)} \\
R + Y &\xrightarrow{\theta_2 RY} & 2R+Y & \text{(birth of predator)} \\
R + Y &\xrightarrow{\theta_3 RY} & R & \text{(death of prey)} \\
Y & \xrightarrow{\mathmakebox[2em]{\theta_4Y}} & 2Y & \text{(birth of prey)}
\end{array}
\end{equation}
The data used by the authors is depicted in the bottom-left panel of \cref{fig:datasets}.
It represents observations at regular intervals with $\Delta t=200$ of a path 
simulated using $\theta_\text{true}=(1,5,5,1)\cdot10^{-4}$.

We compare our methods against the pseudo-marginal sampler
proposed by the authors, which builds a chain directly targeting the augmented posterior for $\theta$ using i.i.d.\ $\text{Gamma}(4,10^4)$ priors. As the third panel in \cref{fig:benchmarks} shows, we achieve an improvement of
3 orders of magnitude using the IA sampler with uniformization. AMH11 denotes the matrix exponentiation algorithm described in \citet{al2011computing}.

\subsection{SSIR model}

The last CTMC we consider is the SSIR model, described in detail in \cref{subsec:reaction_networks}. The data---taken from \citet{georgoulas_unbiased_2017} and shown in the bottom-right panel of \cref{fig:datasets}---are obtained from observations taken at irregular
time steps from a path simulated using $\theta_\text{true}=(0.4, 0.5, 0.4)$. In our experiments, we match the prior used by the authors, corresponding to i.i.d.\ $\text{Gamma}(1.5,5)$. We can see in the last panel of \cref{fig:benchmarks} that the IA sampler with uniformization achieves an improvement of almost 4 orders of magnitude.

\section{Conclusions}

In this work we have proposed pseudo-marginal methods based on monotonic
approximations to the likelihood of stochastic processes with intractable
transition functions. A monotonic approximation can be seen as one of two inputs---the
other being a debiasing scheme---of an almost automatic procedure for designing
pseudo-marginal samplers targeting the correct posterior distribution. To this
end, we have described the OSTE scheme which
improves over other debiasing methods by using a bounded number of calls to the
approximating sequence, and by offering a flexible way to balance the
variance-cost trade-off. We have also shown how to accelerate a recently
proposed pseudo-marginal sampler for countable state-space CTMCs by considering
matrix exponentiation algorithms that produce increasing approximations when
applied to finite rate matrices. Moreover, we have introduced the skeletoid as
a novel probabilistic approach for computing monotonic approximations to
the exponential of rate matrices, that excels in cases where the norm of these
matrices is high. Finally, we have given experimental evidence that our
proposed algorithms offer sizable efficiency gains.

Note that the only component of the proposed algorithms that is tailored to
reaction networks is the design of the increasing sequence of state-spaces.
Therefore, extending our algorithms to CTMCs other than the ones considered in
this study should not be cumbersome. Some interesting areas of application
include CTMCs arising from the study of the secondary structure of nucleic acids
through elementary step kinetic models
\citep{schaeffer2015,zolaktaf2019efficient}, as well as CTMCs underlying
state-dependent speciation and extinction models in phylogenetics
\citep{maddison2007,louca2019}.

More generally, we aim to extend the ``monotone approximation plus debiasing''
framework to other stochastic processes with intractable likelihoods. For example,
exact inference methods for diffusions exist only for a restricted class of
such processes \citep{beskos2006exact,sermaidis2013markov}. On the other hand,
there is a growing literature on asymptotic expansions to the transition
functions of diffusions \citep{ait2008closed,li2013maximum} which are valid in
more general settings. These expansions could potentially be modified
to produce monotone approximations, thus enabling the use of our framework for
exact inference. A different way to tackle this problem would be to rely on
techniques to approximate diffusions using CTMCs
\citep{kushner1980robust,di1981continuous}. As long as the transition function
of the latter converges monotonically to the transition function of the
diffusion, the methods in this paper could be applied.

\bibliography{ref}

\appendix

\section{Implementation details}\label{app:imp_details}

\subsection{Numerically stable implementation of the skeletoid method}

When $\delta=t2^{-k}$ is close to the machine's precision, the straightforward computation of $S(t2^{-k})^{2^k}$ produces rounding errors. The reason is that in these cases $S(\delta)\approx I$. A simple way to bypass this issue is to rely on the identity
\[
A^2 = (I+A-I)^2 = I + 2(A-I) + (A-I)^2
\]
which holds for any square matrix $A$. Defining $B:=A-I$, the identity becomes
\[
A^2 = I+B' \hspace{1em} \text{with} \hspace{1em} B' = 2B + B^2
\] 
An inductive argument then shows that we can compute $A^{2^k}$ via the recursion depicted in \cref{algo:implicit_squaring}.

\begin{algorithm}[t]
\DontPrintSemicolon
\SetKwInOut{Input}{input}
\SetKwInOut{Output}{output}
\Input{$B:=A-I,k$}
\lFor{$i = 1, 2, \dots, k$}{
    $B\gets 2B+B^2$
}
\Output{$A^{2^k}=I+B$.}
\caption{Implicit matrix squaring}
\label{algo:implicit_squaring} 
\end{algorithm}

To fully take advantage of \cref{algo:implicit_squaring} we must implicitly compute $S(\delta)-I$. This is achieved by replacing the formula for the diagonal elements in \cref{eq:def_S_delta} by
\[
[S(\delta)]_{i,i} = \mathtt{expm1}(q_{i,i}\delta)
\]
where \texttt{expm1} is the standard numerical routine for accurately computing $e^x-1$ when $x\approx0$.

\subsection{Computing elements of the matrix exponential}

When computing the likelihood function we are generally interested in only a handful of entries of the transition function. This can be used to  speed-up computation, by only performing computation for the rows of the matrix exponential that contain those entries.

Let $Q_r$ be the rate matrix associated to $\statespace_r$, and define $b_r:=|\statespace_r|$. Suppose we are interested in the entries $\{(i_l,j_l)\}_{l=1}^m$ of the $s$-th approximation to the transition function $M^{(s)}(t)$, with $m\ll b_r$. Define the matrix $L$ of size $m\times b_r$
\[
L_{u,v} :=
\left\{
\begin{array}{cc}
1 & \text{if }v = j_u, \\
0 & \text{otherwise.}
\end{array}
\right.
\]
In the case of uniformization, recall that $P:=I-Q_r/\bar{q}$. Then
\[
(M^{(s)}(t))_{i_l,j_l} = (LM^{(s)}(t))_{l,j_l} = \sum_{n=0}^s e^{\bar{q} t} \frac{(-\bar{q}t)^n}{n!} (LP^n)_{l,j_l}
\]
Note that the matrix $LP^n$ can be efficiently computed recursively setting $LP^0=L$ and updating via $LP^n=(LP^{n-1})P$. Suppose that $P$ has $\eta$ non-zero elements per column on average. Then the update operation involves $\eta b_r$ FLOPs, so we can obtain $(M^{(s)}(t))_{i,j}$ using only $s\eta b_r$ operations. Compare to $s\eta b_r^2$ which would be the cost of computing the full matrix exponential.

For the skeletoid method, we can apply a trick used for computing the action of the 
matrix exponential obtained through scaling and squaring \citep{sherlock2018direct}.
Consider the squaring decomposition
\[
S(\delta)^{2^k} = \underbrace{S(\delta)^{2^{k_1}}\cdot\dots\cdot S(\delta)^{2^{k_1}}}_{2^{k_2} \text{ times}} = \prod_{i=1}^{2^{k_2}} S(\delta)^{2^{k_1}}
\]
for any $k_1,k_2\in\Z_+$ such that $k_1+k_2=k$. Pre-multiplying by $L$
\[
LS(\delta)^{2^k} = L \prod_{i=1}^{2^{k_2}} S(\delta)^{2^{k_1}} = (L S(\delta)^{2^{k_1}}) \prod_{i=1}^{2^{k_2}-1} S(\delta)^{2^{k_1}}
\]
Now, in general $S(\delta)^{2^{k_1}}$ is a dense matrix of size $b_r\times b_r$ (due to CTMCs being irreducible in general), so we model its computational cost to $k_1b_r^3$ (ignoring subcubic algorithms due to their higher associated constants for simplicity). Then the product $L S(\delta)^{2^{k_1}}$ costs $mb_r^2$, and so does every subsequent multiplication by $S(\delta)^{2^{k_1}}$ on the right. Therefore, the total cost becomes
\[
C(k_1,k_2) = \beta k_1b_r^3 + mb_r^2 2^{k_2}
\]
where we have added the parameter $\beta<1$ to account for the fact that $n$ vector-matrix multiplications with cost $n^2$ are slower than one matrix-matrix multiplication with cost $n^3$. We have found $\beta=0.1$ gives consistent results.

We can easily optimize the function above if we work in $\R$ by setting $x\gets k_2$, and $k_1=k-x$. Indeed, let
\[
C(x) = \beta (k-x)b_r^3 + mb_r^2 2^x.
\]
Differentiating gives the condition
\[
0 = C'(x) = -\beta b_r^3 + \log(2) mb_r^2 2^x.
\]
Thus
\[
x^* = \log_2\left( \frac{\beta b_r}{\log(2) m}\right).
\]
Constraining $x^*$ to an integer in the correct interval gives
\[
k_2^* = 0\vee (k\wedge \arg\min\{ C(y): y\in\{\lfloor x^* \rfloor, \lceil x^*\rceil\}\}).
\]
Finally, setting $k_1^*=k-k_2^*$ yields a strategy for computing $LS(\delta)^{2^k}$ which is in general faster than the naive approach requiring forming $S(\delta)^{2^k}$.

\subsection{Building the base state space}\label{app:build_base_statespace}

Here we present a heuristic for computing a ``seed'' path between pairs of 
observations that is specialized to reaction networks with $\statespace=\Z_+^{n_s}$. 
Fix an observed change in the state of the system $(x_{i-1},x_i,\Delta t_i)$, 
with $\Delta t_i := t_i-t_{i-1}$. Assuming non-explosivity, there 
exists a finite sequence of states $\{s_p\}_{p=0}^P$ and a corresponding sequence 
of reactions $\{r_p\}_{p=1}^P$---with $s_0=x_{i-1}$ and $s_P=x_i$---such that for
all $p\in\{1,\dots,P\}$
\[
s_p = s_{p-1} + U_{r_p,\cdot},
\]
where $U$ is the reaction matrix of the model of size $n_r\times n_s$. Unrolling the above recursion yields
\begin{equation}\label{eq:jump_obs_unrolled_reaction}
\Delta x_i := x_i - x_{i-1} = \sum_{p=1}^P U_{r_p,\cdot} = U^T\nu,
\end{equation}
where $\nu\in\Z_+^{n_r}$ is an integer vector denoting the number of times each reaction is used to form the path.

The solution to \cref{eq:jump_obs_unrolled_reaction}, although not necessarily unique, can be computed by expressing the problem as an Integer Program
\[
\begin{array}{ccc}
\displaystyle\min_{\nu \in \Z^{n_r}} \boldsymbol{1}_{n_r}^T\nu \hspace{1em}  & \text{s.t.} & U^T\nu = \Delta x_i,\; \nu \geq 0.
\end{array}
\]
Since for the reaction networks considered in this paper, $n_r,n_s<10$, the problem is computationally tractable. Moreover, since $\Delta x_i$ is an integer vector, when $U$ is Totally Unimodular---which is the case for all the reaction networks we experiment with---the corresponding Linear Program relaxation yields an integer solution \citep{hoffman1956integral}. Finally, once $\nu$ is obtained, we recover the path $\{s_p\}_{p=1}^P$---with $P:=\boldsymbol{1}_{n_r}^T\nu$---using a simple heuristic search algorithm.

\subsection{Numerically stable computation of the log-likelihood}

The formula for computing $\tilde{L}_{\text{RA}}(\theta,U)$ in \cref{eq:def_est_RA} is numerically unstable because it depends on the products in \cref{eq:def_joint_approx_like}, which in most cases will be $\approx 0$. However, we can use the following identities to directly compute $\log(\tilde{L}_{\text{RA}}(\theta,U))$ without explicitly computing products. The function \texttt{expm1} is the standard numerical routine for accurately computing $e^x-1$ when $x\approx0$. Similarly, \texttt{log1p} refers to standard numerical routines that accurately compute $\log(1+x)$ for $|x|\ll 1$.

\begin{proposition}\label{prop:accurate_log_L}
Let $0\leq p_1\leq p_2\leq p_3$. Fix $\alpha\in(0,1]$ and let
\[
z := p_1 + \frac{p_3-p_2}{\alpha}.
\]
Define $s_i := \log(p_i)$ for each $i\in\{1,2,3\}$. Then
\begin{align}
\log(z)=\left\{
\begin{array}{cc}
-\infty & p_3 = 0\\
s_3-\log(\alpha) & p_3 > 0, p_2=0 \\
s_2 + \log(\mathtt{expm1}(s_3-s_2)) -\log(\alpha) & p_2 > 0, p_1=0 \\
s_1 + \mathtt{log1p}\left(\exp(s_2-s_1-\log(\alpha))\mathtt{expm1}(s_3-s_2)\right) & \text{otherwise.}
\end{array}
\right.
\end{align}
\end{proposition}

\begin{proof}
The first two cases follow directly. For the third case, note that
\[
z = \frac{\exp(s_3)-\exp(s_2)}{\alpha} = \exp(s_2)\frac{(\exp(s_3-s_2)-1)}{\alpha} = \exp(s_2)\frac{\mathtt{expm1}(s_3-s_2)}{\alpha}.
\]
Hence
\[
\log(z) = s_2 + \log(\mathtt{expm1}(s_3-s_2)) - \log(\alpha).
\]
The general case follows similarly
\begin{align}
z &= \exp(s_1) + \frac{\exp(s_3)-\exp(s_2)}{\alpha} \\
&= \exp(s_1)\left(1 + \frac{\exp(s_3)-\exp(s_2)}{\alpha \exp(s_1)}\right) \\
&= \exp(s_1)\left(1 + \exp(s_2)\frac{\mathtt{expm1}(s_3-s_2)}{\alpha \exp(s_1)}\right) \\
&= \exp(s_1)\left(1 + \exp(s_2-s_1-\log(\alpha))\mathtt{expm1}(s_3-s_2)\right).
\end{align}
Taking $\log(\cdot)$ at both sides gives the required expression
\[
\log(z) = s_1 + \mathtt{log1p}(\exp(s_2-s_1-\log(\alpha))\mathtt{expm1}(s_3-s_2)).
\]
\end{proof}

\subsection{Tuning the pseudo-marginal samplers}\label{app:tuning_stopping_times}

\begin{figure}[t]
\centering
\includegraphics[width=0.9\textwidth]{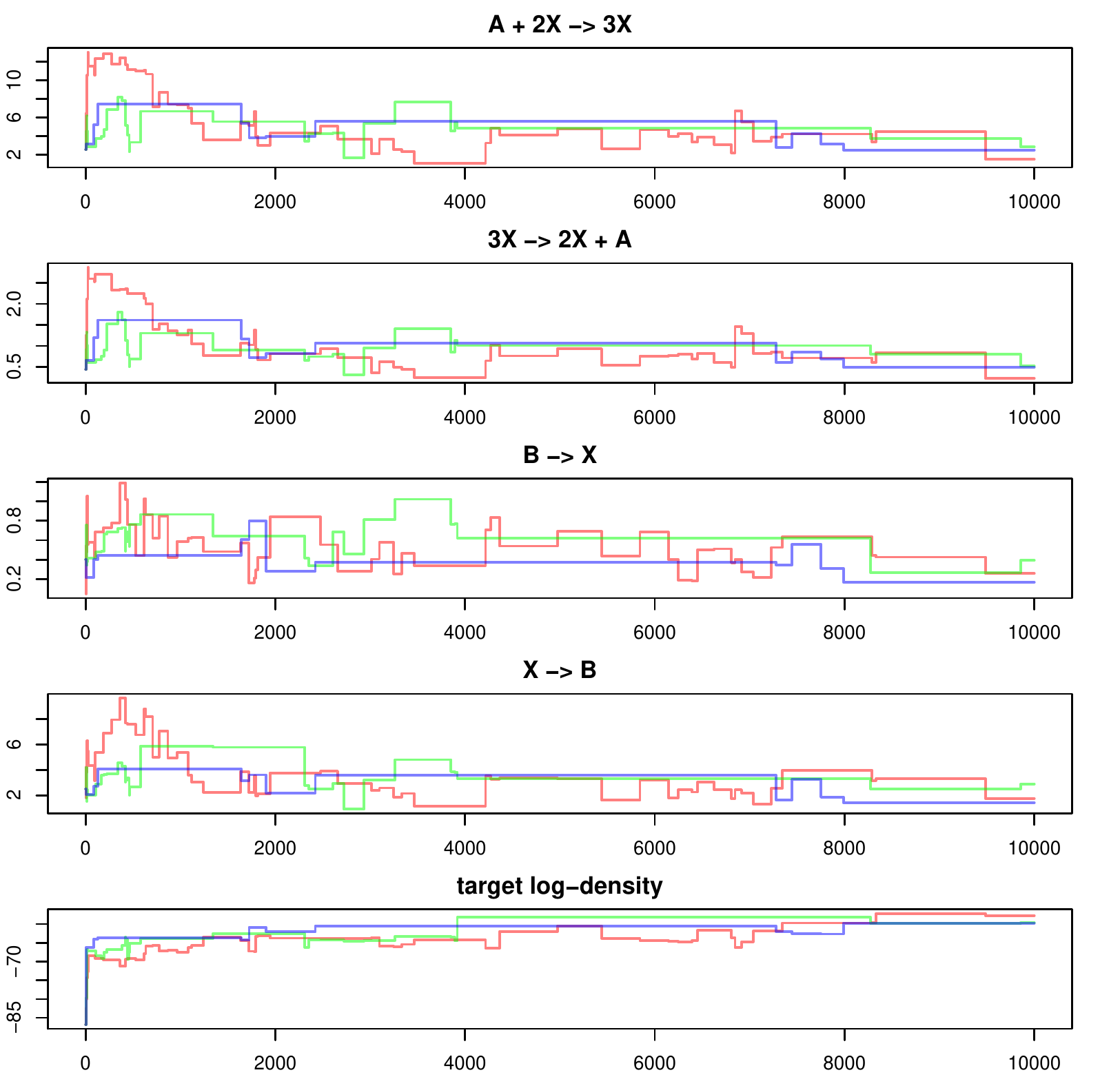}
\caption{Example of a poorly tuned pseudo-marginal sampler for the Schl{\"o}gl model. Each line depicts an independent parallel chain.}
\label{fig:mcmc_fail}
\end{figure}

The approach we have presented depends crucially on correctly tuning the joint truncation--accuracy sequences (see \cref{subsec:monotone_approx_matrix_exp}) and the probability mass function $p(n)$ of the stopping times so that de-biasing behaves properly. As \cref{fig:mcmc_fail} shows, without proper tuning the sampler can give rise to chains that get stuck for long stretches (in stark contrast to \cref{fig:mcmc_traceplots}). The reason for this behavior is to be found in \cref{eq:debias_estimator}. The ratio can give rise to arbitrarily high estimates of the likelihood if the probability mass function $p(n)$ is much smaller than the observed decrease in error (i.e., the numerator). In this respect, the design of $p(n)$ bears similarity with the design of proposal distributions for importance sampling, as we must ensure that the tails of $p(n)$ match the decay of the first differences $a_{\omega+n+1} - a_{\omega+n}$.

There is a growing literature concerned with optimally tuning pseudo-marginal samplers to maximize some measure of efficiency \citep{pitt2012properties,doucet2015efficient,sherlock2015efficiency,schmon2020large}. Recall the augmented target density in \cref{eq:pseudo_marginal_aug_target_density}, and define
\begin{equation}
\zeta(\theta,U) :=  \log(\tilde{L}(\theta,U)) - \log(L(\theta)).
\end{equation}
Moreover, let $\sigma_\zeta^2(\theta) := \var(\zeta|\theta)$. The recommendations suggest tuning the sampler to achieve $\sigma_\zeta(\theta_\text{MAP}) = \bar{\sigma}$, for some pre-specified target $\bar{\sigma}$, and with $\theta_\text{MAP}$ a mode of $\pi(\theta|\D)$. The benefit of this approach is that it does not require running the chains for long periods of time with different configurations, since $\sigma_\zeta(\theta_\text{MAP})$ can be estimated using simple Monte Carlo whenever we can sample $U \sim m$.

On the other hand, depending on the particular assumptions of each study, the recommended $\bar{\sigma}$ ranges from 0.92 to 1.8. Given the ambiguity of these different recommendations, we use $\sigma_\zeta(\theta_\text{MAP})$ not as something to tune in itself, but as a summary of the characteristics of different samplers. In short, we tune our samplers following these steps:
\begin{enumerate}
\item Tune a sampler with a default value of $p_\text{min}=0.9$ (as shown in the next paragraphs) and do a preliminary run to find $\theta_\text{MAP}$ and an estimate $\widehat{\var}(\theta|\D)$ of  $\var(\theta|\D)$.
\item Design an array of different samplers by running the tuning procedure explained in the next paragraphs for all
\[
p_\text{min}\in\{0,0.01,0.1,0.2,0.4,0.6,0.8,0.9\}
\]
\item For each sampler, take $\{U_i\}_{i=1}^{100} \iidsim m$ and estimate $\sigma_\zeta(\theta_\text{MAP})$ using simple Monte Carlo. Also, record the time needed to compute that estimate
\item For each $\bar{\sigma} \in \{0.1,1,1.5,2\}$, find the sampler giving $\sigma_\zeta(\theta_\text{MAP})\leq \bar{\sigma}$ in the shortest time.
\item For each sampler in this shortlist
\begin{enumerate}
\item Optimize $\alpha$ in the variance of the proposal $\Sigma=\alpha\widehat{\var}(\theta|\D)$, aiming to maximize ESS per compute time, using a grid search.
\item Run the sampler with the tuned variance for a longer time and get an estimate of ESS per compute time.
\end{enumerate}
\item Select the sampler giving the highest ESS per compute time.
\end{enumerate}

To implement the general outline described above, let us assume that for every observation $i\in\{1,\dots,n_d\}$ and every $r\in\Z_+$, the sequence $\{\mrk(\Delta t_i)_{x_{i-1},x_i}\}_{k\in\Z_+}$ satisfies a specific version of the second condition in \cref{eq:doublymonotonevariance_assumptions} (\cref{prop:doublymonotonevariance}), namely
\begin{equation}\label{eq:tuning_decay_10_k}
M_r(\Delta t_i)_{x_{i-1},x_i} - \mrk(\Delta t_i)_{x_{i-1},x_i} \leq 10^{-k}.
\end{equation}
Note that we can always satisfy \cref{eq:tuning_decay_10_k} through the use of \cref{eq:def_seq_approx_sol_poisson_quantile} or \cref{eq:def_seq_approx_sol_Skeletoid} (depending on the algorithm). Moreover, since \cref{eq:tuning_decay_10_k} works for $k\in\R_+$ too, we re-index the collection of approximate transition probabilities with $\{\mrk(\Delta t_i)_{x_{i-1},x_i}\}_{r\in\Z_+,k\in\R_+}$. This allows us to have a finer control in the tuning process.

In order to have a generic tuning procedure that works for both IA and RA estimators, let $\ark$ be either
\begin{itemize}
\item the estimated transition probability for the $i$-th observation $\mrk(\Delta t_i)_{x_{i-1},x_i}$, if using method IA (see \cref{eq:def_est_IA_elements}), or
\item the estimated likelihood for the full data $L_r^{(k)}(\theta)$ if using method RA (see \cref{eq:def_joint_approx_like}).
\end{itemize}
Following the conclusions of \cref{prop:doublymonotonevariance}, we parametrize the sequence of pairs $\{(r_n,k_n)\}_{n\in\Z_+}$ using a simple linear form
\begin{equation}\label{eq:def_joint_sequence_linear}
\forall n\in\Z_+: r_n := \omega^{(\text{T})} + n \qquad \text{ and } \qquad k_n := \omega^{(\text{S})} + \sigma n,
\end{equation}
where $\omega^{(\text{T})}\in\Z_+$ is a truncation offset, $\omega^{(\text{S})}\in\R$ an
approximation quality offset, and $\sigma>0$ a relative step-size. Indeed, in \cref{prop:doublymonotonevariance}, the offsets are $0$ and $\sigma=\rho/\kappa$ if one allows for $k_n\in\R_+$ as we do here. We shall see that having non-zero offsets lets us control de variance of the estimators.

Let us begin by setting $\omega^{(\text{T})},\omega^{(\text{S})}$ to preliminary values. Define
\begin{align}
\aseq{r}{\infty} &:= \lim_{k\to\infty} \ark \hspace{1em} \forall r \in \Z_+ \\
\aseq{\infty}{k} &:= \lim_{r\to\infty} \ark \hspace{1em} \forall k \in \R_+ \\
\alpha &:= \lim_{r\to\infty} \aseq{r}{\infty}
\end{align}
$\alpha$ denotes the target quantity to be estimated; that is, either the transition probability $M(\Delta t_i)_{x_{i-1},x_i}$ if using the IA method, or the likelihood $L(\theta)$ if using the RA approach.
\cref{prop:debiasing} shows that we can arbitrarily reduce the variance of $Z$ by increasing the offset $\omega$. On the other hand, the cost of computing $Z$ is at least the cost of computing $a_\omega$, which in our context increases to $\infty$ as $\omega\to \infty$. It follows that there should be a sweet-spot maximizing the efficiency of the sampler. Following this intuition, we would set
\begin{align}
\omega^{(\text{T})} &= \min\left\{r\in\Z_+: \aseq{r}{\infty} \geq p_\text{min}\alpha\right\} \\
\omega^{(\text{S})} &= \min\left\{k\in\R_+: \aseq{\infty}{k} \geq p_\text{min}\alpha\right\}
\end{align}
for some $p_\text{min}\in(0,1)$ which would reflect a trade-off between guaranteed accuracy and cost. In turn, $p_\text{min}$ would be set by maximizing ESS per compute time.

\begin{algorithm}[t]
\DontPrintSemicolon
\SetKwInOut{Input}{input}
\SetKwInOut{Output}{output}
\Input{$\epsilon=10^{-8}, r_{\text{explore}}=15$}
Set accuracy $k_\epsilon$ to achieve error $\epsilon$ using \cref{eq:def_seq_approx_sol_poisson_quantile} or \cref{eq:def_seq_approx_sol_Skeletoid}\;
\Comment{Step 1: analyze truncation $r$ at fixed algorithm setting $k=k_\epsilon$}
$D^{(\text{T})}\gets\emptyset$ \Comment*{storage for truncation profile}
$a_{\text{old}} \gets 0, \Delta \gets \infty, r\gets0$\;
\While{$r<r_{\text{ explore}}$ or $\Delta\geq\epsilon$}{
  $a_{\text{new}} \gets a_r^{k_\epsilon}$\Comment*{get estimate at $(r,k_\epsilon)$ accuracy}
  $D^{(\text{T})} \gets D^{(\text{T})} \cup \{[r,a_{\text{new}}]\}$\Comment*{add to profile}
  $\Delta\gets a_{\text{new}}-a_{\text{old}}$\;
  $a_{\text{old}} \gets a_{\text{new}}$\;
  $r\gets r+1$\;
}
$r_\epsilon \gets r-1$\Comment*{truncation achieving $\Delta<\epsilon$}
$a^* \gets a_{\text{new}}$\Comment*{best guess of exact solution}
\;
\Comment{Step 2: analyze algorithm setting $k$ at fixed truncation $r=r_\epsilon$.}
$D^{(\text{S})}\gets\emptyset$ \Comment*{storage for accuracy profile}
$\Delta \gets 1, k\gets-10$\;
\While{$\Delta\geq\epsilon$}{
  $a_{\text{new}} \gets a_{r_\epsilon}^k$\Comment*{get estimate at $(r_\epsilon,k)$ accuracy}
  $D^{(\text{S})} \gets D^{(\text{S})} \cup \{[k,a_{\text{new}}]\}$\Comment*{add to profile}
  $\Delta\gets a^*-a_{\text{new}}$\;
  $k\gets k+1$\;
}
$k_\epsilon \gets k-1$\Comment*{algorithm setting achieving $\Delta<\epsilon$}
\Output{$D^{(\text{T})},D^{(\text{S})},a^*,r_\epsilon,k_\epsilon$}
\caption{Separate analysis of convergence of sequences}
\label{algo:study_convergence} 
\end{algorithm}

Implementing the above strategy has the obvious drawback that $\alpha$ is unknown. Moreover, we do not have a theoretical bound for the truncation error that would tell us how high to set $r$ to approximate $\alpha$. A similar thing happens with the matrix exponential approximation error when we use the skeletoid method, since its bound can be loose in some settings. We must therefore begin by empirically inspecting the behavior of $\{a_r^\infty\}_{r\in\Z_+}$ and $\{a_\infty^k\}_{k\in\R_+}$. \cref{algo:study_convergence} achieves this by first exploring the truncation sequence when the matrix exponentiation algorithm is set to a fixed high precision setting, producing a truncation index $r_\epsilon$ at which the difference between consecutive solutions is less than a pre-specified threshold $\epsilon>0$. It also gives a best guess for the true value $a^*\approx \alpha$. After this, it sets the truncation to $r_\epsilon$ and proceeds to check the convergence in $k$ by scanning $a_{r_\epsilon}^k$ at integer values of $k$, until it finds $k_\epsilon$ such that $a^* - a_{r_\epsilon}^{k_\epsilon} < \epsilon$. The algorithm returns all these values plus the complete profiles $D^{(\text{T})},D^{(\text{S})}$.

\begin{figure}[t]
\centering
\includegraphics[width=0.9\textwidth]{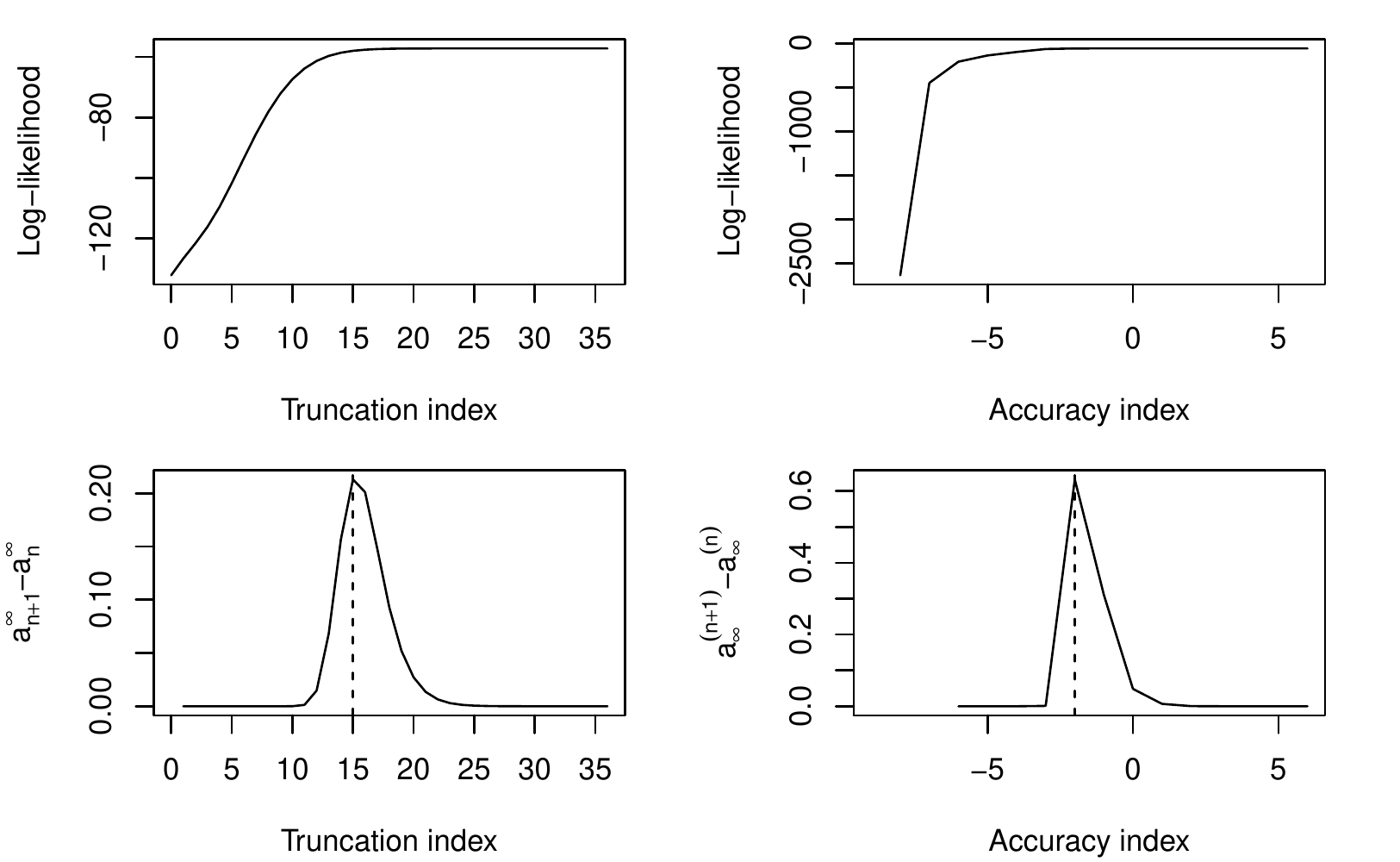}
\caption{Visualization of the tuning process for an RA pseudo-marginal sampler for the Schl{\"o}gl model, with $\theta=\arg\max\pi(\theta|\D)$. \textbf{First row}: output of \cref{algo:study_convergence}. \textbf{Second row}: profiles of the first differences computed from that output.}
\label{fig:study_conv}
\end{figure}

The first row in \cref{fig:study_conv} gives an example of the output for an RA sampler for the Schl{\"o}gl model. We see that truncating at $r = 20$ gives a good approximation, while the skeletoid method gives a very accurate solution for $k=0$, indicating that the associated bound is indeed loose for this problem.

The output from \cref{algo:study_convergence} allows us to proceed with tuning the offsets. As mentioned in \cref{sec:experiments}, we use stopping times with distribution $N\sim\text{Geom}(p)$; i.e., $p(n)=p(1-p)^n$ for all $n\in\Z_+$. Note that this probability mass function is exponentially decreasing with rate $\beta := -\log(1-p)$. On the other hand, we can verify that the distribution given by
\begin{equation}\label{eq:OSTE_zero_var_dist}
\forall n \in\Z_+: p(n) := \frac{a_{\omega+n+1}-a_{\omega+n}}{\alpha-a_{\omega+n}} \propto a_{\omega+n+1}-a_{\omega+n}
\end{equation}
satisfies all the requirements in \cref{prop:debiasing} and gives $\var(Z)=0$. Aiming to mimic this optimal distribution using a $\text{Geom}(p)$ imposes the requirement that the sequence of first differences $\Delta a_{\omega+n}:=a_{\omega+n+1}-a_{\omega+n}$ be strictly decreasing. This is in general not true for any offset, as the second row in \cref{fig:study_conv} shows. Indeed, one can observe peaks at $r_\text{peak}=15$ and $k_\text{peak}=-2$ for truncation and matrix exponential accuracy respectively. Thus, we impose that offsets be set to at least the last peak observed for the corresponding $\{\Delta a_{\omega+n}\}_{n\geq0}$ profiles
\begin{equation}\label{eq:tuning_prelim_offsets}
\begin{aligned}
\omega^{(\text{T})} &= r_\text{peak} \vee \min\{r: (r,a_r^{k_\epsilon}) \in D^{(\text{T})} \text{ and }  a_r^{k_\epsilon} \geq p_\text{min}a^* \} \\
\omega^{(\text{S})} &= k_\text{peak} \vee \min\{k: (k,a_{r_\epsilon}^k) \in D^{(\text{S})} \text{ and } a_{r_\epsilon}^k \geq p_\text{min}a^* \}.
\end{aligned}
\end{equation}

\begin{algorithm}[t]
\DontPrintSemicolon
\SetKwInOut{Input}{input}
\SetKwInOut{Output}{output}
\Input{$r_\epsilon,k_\epsilon,\omega^{(\text{S})}, \omega^{(\text{T})},\sigma_\text{default}=0.1$}
$\sigma\gets \sigma_\text{default} \vee \left(\frac{r_\epsilon - \omega^{(\text{T})}}{k_\epsilon - \omega^{(\text{S})}}\right)$ \Comment*{initialize with \cref{eq:tuning_sigma_init}}
$a_{\text{old}} \gets a_{\omega^{(\text{T})}}^{\omega^{(\text{S})}}$\Comment*{initialize with solution at both offsets}
$n\gets 0; r \gets \omega^{(\text{T})} ;\ k \gets \omega^{(\text{S})}$\Comment*{initialize sequences}
\While{$r\leq r_\epsilon$ and $k\leq k_\epsilon$}{
  $r \gets \omega^{(\text{T})} + n ;\ k \gets \omega^{(\text{S})} + \sigma n$\;
  $a_{\text{new}} \gets a_{r}^{k}$\Comment*{get estimate at $(r,k)$ accuracy}
  \While(\Comment*[f]{monotonicity fails}){$a_{\text{old}}>a_{\text{new}}$ and $k\leq k_\epsilon$}{
    $\sigma\gets 2\sigma$\Comment*{double $\sigma$ until monotonicity holds}
    $k \gets \omega^{(\text{S})} + \sigma n$\;
    $a_{\text{new}} \gets a_{r}^{k}$\Comment*{update estimate}
  }
  $a_{\text{old}}\gets a_{\text{new}}$\;
  $n\gets n+1$\;
}
\Output{$\sigma$}
\caption{Tune $\sigma$}
\label{algo:set_AST_sigma} 
\end{algorithm}

Having set preliminary values for $\omega^{(\text{S})},\omega^{(\text{T})}$, we proceed to set $\sigma$. Following the intuition from \cref{prop:doublymonotonevariance}, we initialize this parameter with the rule
\begin{equation}\label{eq:tuning_sigma_init}
\sigma=\sigma_\text{default} \vee \left(\frac{r_\epsilon - \omega^{(\text{T})}}{k_\epsilon - \omega^{(\text{S})}}\right).
\end{equation}
The ratio on the right-hand side corresponds to an ``average relative velocity'' of the convergence of the matrix exponential accuracy sequence versus the truncation sequence, while $\sigma_\text{default}>0$ corrects border cases. For doubly monotone sequences, there is no additional tuning required. In other cases, we must set $\sigma$ high enough to ensure monotonicity holds. The process is summarized in \cref{algo:set_AST_sigma}. Whenever we encounter a failure in monotonicity $\aseq{r_{n+1}}{k_{n+1}} < \aseq{r_{n}}{k_{n}}$, we double $\sigma$ and recompute $\aseq{r_{n+1}}{k_{n+1}}$ until the condition is met.

Having set $\sigma$ we can adjust the offsets to match the behavior of the joint sequence. Let $n_\text{peak}$ be the index of the rightmost peak in the sequence of differences $\Delta a_{n+1} := \aseq{r_{n+1}}{k_{n+1}}-\aseq{r_{n}}{k_{n}}$. Define
\[
n_\text{offset} := \min\{n\geq  n_\text{peak}: \aseq{r_{n}}{k_{n}}\geq p_\text{min}a^*\}.
\]
and update the offsets to $\omega^{(\text{T})}=r_{n_\text{offset}}$ and $\omega^{(\text{S})}=k_{n_\text{offset}}$.

Finally, to tune $p_g$ we fit a linear regression without intercept to
\[
\log(\Delta a_n)-\log(\Delta a_0) = \beta_g n + \varepsilon_n
\]
where $\{\epsilon_i\}\overset{i.i.d.}{\sim}\mathcal{N}(0,\sigma_\epsilon^2)$. Then we set $p_g = \underline{p} \vee (\overline{p} \wedge (1-e^{\widehat{\beta_g}}))$, where $[\underline{p},\overline{p}]$ defines a pre-specified safe interval (we use $\underline{p}=0.4$ and $\overline{p}=0.9$).

\section{Proofs}\label{app:proofs}

\begin{proof}[\textbf{Proof of \cref{prop:debiasing}}]

The fact that $Z\geq a_\omega$ almost surely follows directly from the definition 
of $Z$ and the non-decreasing 
property of the sequence.
Unbiasedness follows by the monotone convergence theorem:
\begin{align}
\E[Z] &= a_{\omega} + \E\left[\sum_{n: p(n)>0} \I\{N=n\} \frac{(a_{\omega+n+1} - a_{\omega+n})}{p(n)}\right] \\
&= a_{\omega} + \sum_{n: p(n)>0} p(n) \frac{(a_{\omega+n+1} - a_{\omega+n})}{p(n)} \\
&= a_{\omega} + \sum_{n: p(n)>0} (a_{\omega+n+1} - a_{\omega+n}) \\
&= a_{\omega} + \sum_{n\geq 0} (a_{\omega+n+1} - a_{\omega+n}) && (\text{\cref{eq:debias_estimator_assu_domination}})\\
&= \lim_{n\to\infty}a_{n} \\
&=\alpha.
\end{align}
Let $Y:=Z-a_{\omega}$, so that
\[
Y = \frac{(a_{\omega+N+1} - a_{\omega+N})}{p(N)} = \sum_{n:p(n)>0} \I\{N=n\} \frac{(a_{\omega+n+1} - a_{\omega+n})}{p(n)}.
\]
Then
\begin{align}
Y^2 &= \left(
\sum_{n:p(n)>0} \I\{N=n\} \frac{(a_{\omega+n+1} - a_{\omega+n})}{p(n)}
\right)^2 \\
&=
\sum_{n: p(n)>0} \I\{N=n\} \frac{(a_{\omega+n+1} - a_{\omega+n})^2}{p(n)^2}.
\end{align}
Applying expectations on both sides and using the monotone convergence theorem,
\begin{align}
\E[Y^2] &= \E\left[\sum_{n: p(n)>0} \I\{N=n\} \frac{(a_{\omega+n+1} - a_{\omega+n})^2}{p(n)^2}\right] \\
&= \sum_{n: p(n)>0} \P(N=n) \frac{(a_{\omega+n+1} - a_{\omega+n})^2}{p(n)^2} \\
&= \sum_{n: p(n)>0} \frac{(a_{\omega+n+1} - a_{\omega+n})^2}{p(n)}.
\end{align}
 Finally, the expression for $\var(Z)$ follows from
\begin{align}
\var(Z) &= \var(Y) \\
&= \E[Y^2] - \E[Y]^2 \\
&= \sum_{n: p(n)>0} \frac{(a_{\omega+n+1} - a_{\omega+n})^2}{p(n)} - (\alpha-a_\omega)^2.
\end{align}
If the variance is finite, then the series above is finite; the fact that $a_\omega \uparrow \alpha$ yields $\var(Z) \to 0$ as $\omega\to\infty$.
\end{proof}

\begin{proof}[\textbf{Proof of \cref{prop:suff_cond_debiasing}}]
Since $p(n)>0$ for all $n\in\Z_+$, by \cref{prop:debiasing},
\begin{align}
 \var(Z) &= \sum_{n\in\Z_+} \frac{(a_{\omega+n+1} - a_{\omega+n})^2}{p(n)} - (\alpha-a_\omega)^2\\
 &\leq \sum_{n\in\Z_+} \frac{(a_{\omega+n+1} - a_{\omega+n})^2}{p(n)}\\
 &\leq \sum_{n\in\Z_+} \frac{c^2e^{-2(\omega+n)r}}{e^{-n\beta}(1-e^{-\beta})}\\
 &= \frac{c^2e^{-2\omega \rho}}{1-e^{-\beta}} \sum_{n\in\Z_+} e^{-n(2\rho-\beta)}.
\end{align}
The result follows from the geometric series formula.
\end{proof}

\begin{proof}[\textbf{Proof of \cref{prop:doubly_monotone_sequences}}]
	From the monotonicity assumptions, $\aseq{r_n}{k_n}$ converges to some limit $\beta \le \alpha$. Suppose for the purpose of contradiction that  $\beta < \alpha$. 
	
	From the definition of $a_r$ and the assumption $r_n \uparrow \infty$, we can  pick $n_0 \in\Z_+, \epsilon > 0$ such that
	\begin{align}\label{eq:n-selection}
	a_{r_{n_0}} > \beta + \epsilon.
	\end{align} 
	Then, for $j \in\N$, iteratively pick $n_j > n_{j-1}$ such that $k_{n_j} > k_{n_{j-1}}$. This is possible using the assumption $k_n \uparrow \infty$. Then, 
	\begin{align}\label{eq:vertical-limit} 
    \aseq{r_{n_0}}{k_{n_j}} \uparrow a_{r_{n_0}}.
	\end{align}
	Since $n_j > n_0$, and $r_n$ is increasing, $r_{n_j} \ge r_{n_0}$, and hence 
	\begin{align}\label{eq:bound-projection}
	\aseq{r_{n_j}}{k_{n_j}} \ge \aseq{r_{n_0}}{k_{n_j}}.
	\end{align}
	Combining \cref{eq:n-selection,eq:vertical-limit,eq:bound-projection} yields
	\begin{align}\label{eq:combined}
	\beta=\lim_{j\to\infty}\aseq{r_{n_j}}{k_{n_j}} &\ge \lim_{j\to\infty} \aseq{r_{n_0}}{k_{n_j}} = a_{r_{n_0}} > \beta + \epsilon,
	\end{align}
	which is a contradiction. We conclude that $\beta=\alpha$. 
	
\end{proof}

\begin{proof}[\textbf{Proof of \cref{prop:doublymonotonevariance}}]
Since $r_n\uparrow\infty$ and $k_n\uparrow\infty$, \cref{prop:doubly_monotone_sequences} shows that $\aseq{r_n}{k_n}\uparrow\alpha$. Moreover,
\begin{align}
\aseq{r_{n+1}}{k_{n+1}} - \aseq{r_n}{k_n} &\leq \alpha - \aseq{r_n}{k_n} \\
&=(\alpha - a_{r_n}) + (a_{r_n} - \aseq{r_n}{k_n}) \\
&\leq c_1e^{-\rho n} + c_2e^{-\kappa \lceil(\rho/\kappa) n\rceil} \\
&\leq c_1e^{-\rho n} + c_2e^{-\kappa (\rho/\kappa) n} \\
&\leq 2\max\{c_1,c_2\}e^{-\rho n}.
\end{align}
\end{proof}

\begin{lemma}\label{lemma:conservatization}
Let $Q$ be a non-conservative rate matrix. Define
\begin{equation}\label{eq:def_conservatization}
R := 
\begin{pmatrix}
0 & \boldsymbol{0}^T\\
-Q\boldsymbol{1} & Q
\end{pmatrix}.
\end{equation}
Then $R$ is conservative, with $\inf_{i}R_{x,x} = \inf_{i}Q_{x,x}$, and for $n\in\N$,
\begin{equation}\label{eq:conservatization_powers}
R^n = 
\begin{pmatrix}
0 & \boldsymbol{0}^T\\
-Q^n\boldsymbol{1} & Q^n
\end{pmatrix}.
\end{equation}
Moreover, for any $t\geq0$
\begin{equation}\label{eq:conservatization_matexp}
e^{tR} = 
\begin{pmatrix}
0 & \boldsymbol{0}^T\\
\boldsymbol{1}-e^{tQ}\boldsymbol{1} & e^{tQ}
\end{pmatrix}.
\end{equation}
\end{lemma}
\begin{proof}
The first property $\inf_{i}R_{x,x} = \inf_{i}Q_{x,x}$ is clear from the definition of R in \cref{eq:def_conservatization}, since only a $0$ is added to the diagonal, whose elements are always non-positive. For the second property, we proceed by induction. The case $n=2$ gives
\[
R^2 = 
\begin{pmatrix}
0 & \boldsymbol{0}^T\\
-Q\boldsymbol{1} & Q
\end{pmatrix}
\begin{pmatrix}
0 & \boldsymbol{0}^T\\
-Q\boldsymbol{1} & Q
\end{pmatrix}
=
\begin{pmatrix}
0 & \boldsymbol{0}^T\\
-Q^2\boldsymbol{1} & Q^2
\end{pmatrix}.
\]
Assuming the property holds for some $n$, then
\[
R^{n+1} = R^n R =
\begin{pmatrix}
0 & \boldsymbol{0}^T\\
-Q^n\boldsymbol{1} & Q^n
\end{pmatrix}
\begin{pmatrix}
0 & \boldsymbol{0}^T\\
-Q\boldsymbol{1} & Q
\end{pmatrix}
=
\begin{pmatrix}
0 & \boldsymbol{0}^T\\
-Q^{n+1}\boldsymbol{1} & Q^{n+1}
\end{pmatrix}.
\]
We conclude the result holds for all $n\in\N$. Next, note that
\begin{align}
e^{tR} &= I + \sum_{n = 1}^\infty \frac{t^n}{n!}R^n \\
&= I +
\begin{pmatrix}
0 & \boldsymbol{0}^T\\
-\sum_{n = 1}^\infty \frac{t^n}{n!}Q^n\boldsymbol{1} & \sum_{n = 1}^\infty \frac{t^n}{n!}Q^n
\end{pmatrix} \\
&=
\begin{pmatrix}
1 & \boldsymbol{0}^T \\
- \left[\sum_{n = 1}^\infty \frac{t^n}{n!}Q^n\right]\boldsymbol{1} & e^{tQ}
\end{pmatrix} \\
&=
\begin{pmatrix}
1 & \boldsymbol{0}^T \\
- \left[e^{tQ} - I\right]\boldsymbol{1} & e^{tQ}
\end{pmatrix} \\
&=
\begin{pmatrix}
1 & \boldsymbol{0}^T\\
\boldsymbol{1}-e^{tQ}\boldsymbol{1} & e^{tQ}
\end{pmatrix}.
\end{align}
\end{proof}

\begin{proposition}\label{prop:matexp_nonconservative_rate_mat}
For $r\in\Z_+$, let $Q_r$ be the non-conservative rate matrix given by \cref{prop:monotone_trunc}, with associated truncated state-space $\statespace_{r}$. Then for all $t\geq0$ and $x,y\in\statespace_{r}$,
\[
[e^{tQ_r}]_{x,y} = \P_x(X(t)=y\text{ and }\forall s \in [0,t]: X(s)\in\statespace_{r}),
\]
where $\{X(t)\}_{t\geq0}$ is the CTMC on $\statespace$ driven by the full non-explosive rate matrix $Q$ and initialized at $X(0) = x$.
\end{proposition}
\begin{proof}
Let $\widetilde{\statespace_r} = \statespace_r\cup\{\Delta\}$, and consider expanding $Q_r$ to $\widetilde{Q}_r$ as in \cref{lemma:conservatization}, so that $\Delta$ becomes the absorbing state. Let $\{\widetilde{X}(t)\}_{t\geq0}$ be a CTMC defined on $\widetilde{\statespace_r}$ driven by $\widetilde{Q}_r$ and initialized at $\widetilde{X}(0)=x$. By \cref{lemma:conservatization}, we obtain that for all $x,y\in\statespace_r$
\[
[e^{tQ_r}]_{x,y} = [e^{t\widetilde{Q}_r}]_{x,y} = \P_x(\widetilde{X}(t)=y)
\]
Now, we can couple $(X(t),\widetilde{X}(t))$ by sharing the random variables used in their simulation (see \cref{algo:gillespie}). Denote
\[
\tau_\Delta := \inf\{t\geq 0: \widetilde{X}(t)=\Delta\}
\]
the first hitting time of $\Delta$. Then
\[
\forall t\in[0,\tau_\Delta): X(t)=\widetilde{X}(t)\in\statespace_{r}, \hspace{1em} \text{ and } \hspace{1em} \forall t\geq \tau_\Delta: \widetilde{X}(t)=\Delta.
\]
Therefore,
\begin{align}
[e^{tQ_r}]_{x,y} &= \P_x(\widetilde{X}(t)=y) \\
&= \P_x(\widetilde{X}(t)=y, \tau_\Delta>t) \\
&= \P_x(X(t)=y, \tau_\Delta>t) \\
&= \P_x(X(t)=y\text{ and }\forall s \in [0,t]: X(s)\in\statespace_{r}).
\end{align}
\end{proof}

\begin{proof}[\textbf{Proof of \cref{prop:doubly_monotone_global_unif}}]
To show double monotonicity, we must prove that for all $r,s \in \Z_+$, $t\geq 0$, and $x,y\in\statespace_{r}$,
\begin{equation}
(\mseq{r}{s}(t))_{x,y} \leq (\mseq{r+1}{s}(t))_{x,y}.
\end{equation}
We first show by induction that for all $r\in\Z_+$, $x,y\in\statespace_{r}$, and $n\in\N$,
\[
(P_{r+1}^n)_{x,y} \geq (P_{r}^n)_{x,y}.
\]
For $n=1$,
\[
(P_{r+1})_{x,y} = \delta_{x,y} - \bar{q}^{-1}(Q_{r+1})_{x,y} = \delta_{x,y} - \bar{q}^{-1}(Q_{r})_{x,y} = (P_{r})_{x,y}.
\]
Now suppose that the result holds for some $n\in\N$. Then
\begin{align}
(P_{r+1}^{n+1})_{x,y} &= \sum_{x\in\statespace_{r+1}} (P_{r+1}^n)_{x,z}(P_{r+1})_{z,y} \\
&\geq \sum_{z\in\statespace_{r}} (P_{r+1}^n)_{x,z}(P_{r+1})_{z,y} \\
&\phantom{=}(\text{Discard non-negative summands})\\
&\geq \sum_{z\in\statespace_{r}} (P_{r}^n)_{x,z}(P_{r})_{z,y} \\
&\phantom{=}(\text{Induction and case }n=1) \\
&=(P_{r}^{n+1})_{x,y}.
\end{align}
Thus, the result holds for all $n\in\N$. Finally,
\begin{align}
(\mseq{r+1}{s}(t))_{x,y} &= \sum_{n=0}^s e^{\bar{q} t} \frac{(-\bar{q} t)^n}{n!} [P_{r+1}^n]_{x,y} \geq \sum_{n=0}^s e^{\bar{q} t} \frac{(-\bar{q} t)^n}{n!} [P_{r}^n]_{x,y} = (\mseq{r}{s}(t))_{x,y}.
\end{align}
\end{proof}

\begin{proof}[\textbf{Proof of \cref{prop:unif_ell_infty_error}}]
Note that $P=I-Q/\bar{q}$ is sub-stochastic. Moreover, for all $n\in\Z_+$, $P^n$ is sub-stochastic too. Then,
\begin{align}
0\leq \sum_{y\in\statespace}([M(t)]_{x,y} - [M^{(s)}(t)]_{x,y}) &= \sum_{y\in\statespace}\sum_{n=s+1}^\infty e^{\bar{q} t} \frac{(-\bar{q} t)^n}{n!} [P^n]_{x,y} \\
&= \sum_{n=s+1}^\infty e^{\bar{q} t} \frac{(-\bar{q} t)^n}{n!} \sum_{y\in\statespace}[P^n]_{x,y}\\
&\leq \sum_{n=s+1}^\infty e^{\bar{q} t} \frac{(-\bar{q} t)^n}{n!} 1 \\
&= 1-F(s;\lambda).
\end{align}
The interchange of sums is justified by Tonelli's theorem since the summands are non-negative. Taking the supremum over $i$ gives the required result.
\end{proof}

\begin{lemma}\label{lemma:S_delta_construction}
Let $N_\delta$ be the number of jumps of the CTMC in $(0,\delta]$. Then, for all $x,y\in\statespace$,
\begin{equation}
\P_x(X(\delta)=y,N_\delta\leq 1) = (S(\delta))_{x,y}
\end{equation}
for $S(\delta)$ defined in \cref{eq:def_S_delta}.
\end{lemma}
\begin{proof}
Case $y=x$ is given by
\[
\P_x(N_\delta=0)=\P_x(\text{first jump occurs after }\delta)=e^{q_{x,x}\delta}.
\]
For $y\neq x$,
\begin{align}
\P_x(X(\delta)=y,N_\delta\leq 1) &= \P_x(X(\delta)=y,N_\delta = 1) \\
&=\int_0^\delta \P_x(X(\delta)=y,T_1\in \ud u,T_2>\delta) \\
&= \int_0^\delta (-q_{x,x})e^{q_{x,x}u}\frac{q_{x,y}}{(-q_{x,x})}e^{q_{y,y}(\delta-u)} \ud u \\
&= q_{x,y}e^{q_{y,y}\delta} \int_0^\delta e^{(q_{x,x}-q_{y,y})u} \ud u.
\end{align}
If $q_{x,x}=q_{y,y}$ we get $\P_x(X(\delta)=y,N_\delta\leq 1)=q_{x,y}\delta e^{q_{x,x}\delta}$. Otherwise
\begin{align}
\P_x(X(\delta)=y,N_\delta\leq 1)
&= q_{x,y}e^{q_{y,y}\delta}\frac{1}{q_{x,x}-q_{y,y}} \left. e^{(q_{x,x}-q_{y,y})u} \right|_0^\delta  \\
&= q_{x,y}e^{q_{y,y}\delta}\frac{e^{(q_{x,x}-q_{y,y})\delta}-1}{q_{x,x}-q_{y,y}}   \\
&= q_{x,y}\frac{e^{q_{x,x}\delta}-e^{q_{y,y}\delta}}{q_{x,x}-q_{y,y}}.
\end{align}
\end{proof}

\begin{lemma}\label{lemma:CK_S_delta}
Fix $k\in\Z_+$ and $\delta=t2^{-k}$. Let
\[
N_{l,k} := \text{number of jumps in }(l\delta,(l+1)\delta], \hspace{1em} l\in\{0,1,\dots,2^{k-1}\}.
\]
Then for $x,y \in\statespace$ and $L\in\N$,
\begin{equation}\label{eq:skeletoid_Chapman_Kolmogorov}
\P_x\left(\{X(L\delta)=y\} \cap \bigcap_{l=0}^{L-1} \{N_{l,k}\leq 1\}\right) = [S(\delta)^{L}]_{x,y}.
\end{equation}
\end{lemma}

\begin{proof}
If $Q$ is non-conservative, append additional entries corresponding to an absorbing state $\Delta$ as in \cref{eq:def_conservatization}. Since $\Delta$ is absorbing, all transitions between $x,y\in\statespace$ must necessarily avoid it. Let $\mathcal{F}_t := \sigma(X(s), 0\leq s \leq t)$ be the natural filtration of the process. We can derive the following recursion
\begin{align}
&\P_x\left(\{X((L+1)\delta)=y\} \cap \bigcap_{l=0}^{L} \{N_{l,k}\leq 1\}\right) \\
&=\sum_{z\in\statespace} \E_x\left[\I\{X(L\delta)=z\} \prod_{l=0}^{L-1} \I\{N_{l,k}\leq 1\} \I\{X((L+1)\delta)=y\}  \I\{N_{L,k}\leq 1\}\right] \\
&=\sum_{z\in\statespace} \E_x\left[\I\{X(L\delta)=z\} \prod_{l=0}^{L-1} \I\{N_{l,k}\leq 1\} \E_x\left[\I\{X((L+1)\delta)=y\}  \I\{N_{L,k}\leq 1\}\middle|\mathcal{F}_{L\delta}\right]\right] \\
&=\sum_{z\in\statespace} \E_x\left[\I\{X(L\delta)=z\} \prod_{l=0}^{L-1} \I\{N_{l,k}\leq 1\} \E_z\left[\I\{X(\delta)=y\}  \I\{N_{0,k}\leq 1\}\right]\right] \\
&\phantom{=}(\text{Markov property}) \\
&=\sum_{z\in\statespace} \P_x\left(\{X(L\delta)=z\} \cap \bigcap_{l=0}^{L-1} \{N_{l,k}\leq 1\}\right) \P_z\left(\{X(\delta)=y\} \cap \{N_{0,k}\leq 1\}\right)\\
&\phantom{=}(\text{Independence}) \\
&=\sum_{z\in\statespace} \P_x\left(\{X(L\delta)=z\} \cap \bigcap_{l=0}^{L-1} \{N_{l,k}\leq 1\}\right) [S(\delta)]_{z,y} \\
&\phantom{=}(\text{\cref{lemma:S_delta_construction}})
\end{align}
Now proceed to prove \cref{eq:skeletoid_Chapman_Kolmogorov} by induction. Using $L=1$ in the recursion above yields
\begin{align}
\P_x(X(2\delta)=y, N_{0,k}\leq 1, N_{1,k}\leq 1) &= \sum_{z\in\statespace} \P_x(X(\delta)=z, N_{0,k}\leq 1)S(\delta)_{z,y} \\
&=\sum_{z\in\statespace} S(\delta)_{x,z} S(\delta)_{z,y} \\
&\phantom{=}(\text{Definition of }S(\delta)) \\
&= [S(\delta)^2]_{x,y}.
\end{align}
Assume \cref{eq:skeletoid_Chapman_Kolmogorov} holds for all natural numbers up to $L$. Using $L+1$ in the recursion above yields
\begin{align}
&\P_x\left(\{X((L+1)\delta)=y\} \cap \bigcap_{l=0}^{L} \{N_{l,k}\leq 1\}\right) \\
&=\sum_{z\in\statespace} \P_x\left(\{X(L\delta)=z\} \cap \bigcap_{l=0}^{L-1} \{N_{l,k}\leq 1\}\right) S(\delta)_{z,y}\\
&=\sum_{z\in\statespace} [S(\delta)^L]_{x,z} [S(\delta)]_{z,y} \\
&\phantom{=}(\text{Induction}) \\
&=[S(\delta)^{L+1}]_{x,y}.
\end{align}
\end{proof}

\begin{lemma}\label{lemma:almost_sure_convergence}
Let
\begin{equation}\label{eq:def_Ak}
A_k := \bigcap_{l=0}^{2^k-1} \{N_{l,k} \le 1\}
\end{equation}
be the event which prescribes that at most one jump occurs in each bin. 
If the process is non-explosive, then for all $i\in \statespace$
\begin{equation}\label{eq:almost_sure_convergence}
\P_x\left(\bigcup_{k\in\Z_+} A_k \right) = 1.
\end{equation}
\end{lemma}

\begin{proof}
Let $T$ denote the total number of jump events. For every $n\in\Z_+$, let $B_n = \{T\leq n\}$.

We first show that $B_n\subset \cup_k A_k$. Suppose $\omega \in B_n$, and let $\epsilon_\omega$ denote the smallest inter-arrival time for the jump events in $\omega$. Since $\omega \in B_n$, the number of jump events is finite, therefore $\epsilon_\omega > 0$. Hence there exists a $k_\omega$, namely
\[
t2^{-k_\omega} < \epsilon_\omega \implies k_\omega = \left\lceil\log_2\left(\frac{t}{\epsilon_\omega}\right) \right\rceil,
\]
such that for $k > k_\omega$, $\omega \in A_k$ (to see why, note that the smallest inter-arrival time in $\omega$ is larger than the mesh size in $A_k$, hence there can be at most one event in each block of $A_k$). It follows that $\omega \in \cup_kA_k$.

Now, by non-explosivity, $\P_x(\cup B_n) = 1$. Additionally, $\{B_n\}_{n\in\Z_+}$ is an increasing family of sets. Hence,
\begin{align}
\P_x\left(\cup_k A_k \right) &= \P_x\left(\left[\cup_k A_k\right] \cap \left[\cup_{n} B_n\right] \right) \\
&= \lim_{n\to\infty} \P_x([\cup_k A_k] \cap B_n) && (\text{continuity from below})\\
&= \lim_{n\to\infty} \P_x(B_n) = \P_x(\cup B_n) = 1.
\end{align}
\end{proof}

\begin{proof}[\textbf{Proof of \cref{prop:Skeletoid_monotonicity}}]

Note that $\{A_k\}_{k\in\Z_+}$---as defined in \cref{lemma:almost_sure_convergence}---is an increasing sequence of sets because the $k$-th grid is strictly contained in the $(k+1)$-th grid, so that
\begin{align}
(tl2^{-k},t(2l+1)2^{-(k+1)}] \cup (t(2l+1)2^{-(k+1)},t(l+1)2^{-k}] 
= (tl2^{-k},t(l+1)2^{-k}].
\end{align}
Applying \cref{lemma:CK_S_delta} with $L=2^k$ yields
\begin{equation}\label{eq:squaring_S_delta}
\P_x(X(t) = j, A_k) = [S(\delta)^{2^k}]_{x,y}.
\end{equation}
Since $A_k$ is increasing, $\{X(t)=y\}\cap A_k$ is also increasing. Therefore, by continuity from below
\begin{align}
\lim_{k\to\infty} [S(t2^{-k})^{2^k}]_{x,y} &= \lim_{k\to\infty} \P_x(\{X(t)=y\} \cap A_k) \\
&= \P_x\left(\{X(t)=y\}\cap \bigcup_{k\in\N} A_k\right) \\
&= \P_x(X(t)=y) - \P_x\left(\{X(t)=y\}\cap \left[\bigcup_{k\in\N} A_k\right]^c\right)\\
&= \P_x(X(t)=y) \\
&\phantom{=} \left(\P_x\left(\left[\bigcup_{k\in\N} A_k\right]^c \right) = 0 \text{ by \cref{lemma:almost_sure_convergence}}\right) \\
&= [e^{tQ}]_{x,y}.
\end{align}
\end{proof}

\begin{proof}[\textbf{Proof of \cref{prop:doubly_monotone_skeletoid}}]
To show double monotonicity, we must prove that for all $t\geq0$, $r,s\in\Z_+$, and all $x,y \in \statespace_{r}$,
\begin{equation}\label{eq:skeletoid_increasing_in_r_fixed_s}
(S_r(t2^{-s})^{2^s})_{x,y} \leq (S_{r+1}(t2^{-s})^{2^s})_{x,y}.
\end{equation}
By induction on $s\in\Z_+$. The base case $s=0$ follows from the definition of $Q_r$, since for all $x,y \in \statespace_{r}$
\[
(Q_{r})_{x,y} = Q_{x,y} = (Q_{r+1})_{x,y}.
\]
Thus, for all $t\geq 0$ and $x,y \in \statespace_{r}$
\[
[S_{r}(t)]_{x,y} = [S_{r+1}(t)]_{x,y}.
\]
Now, assume \cref{eq:skeletoid_increasing_in_r_fixed_s} holds for $s\in\Z_+$. Then
\begin{align}
[S_{r+1}(t2^{-(s+1)})^{2^{(s+1)}}]_{x,y} &= [S_{r+1}(t2^{-(s+1)})^{2^s} S_{r+1}(t2^{-(s+1)})^{2^s}]_{x,y} \\
&= \sum_{l\in\statespace_{r+1}} [S_{r+1}(t2^{-(s+1)})^{2^s}]_{x,l} [S_{r+1}(t2^{-(s+1)})^{2^s}]_{l,y} \\
&\geq \sum_{l\in\statespace_{r}} [S_{r+1}(t2^{-(s+1)})^{2^s}]_{x,l} [S_{r+1}(t2^{-(s+1)})^{2^s}]_{l,y} \\
&\phantom{=}(\text{drop non-negative terms for }l\in\statespace_{r+1}\setminus\statespace_{r}) \\
&= \sum_{l\in\statespace_{r}} [S_{r+1}(t'2^{-s})^{2^s}]_{x,l} [S_{r+1}(t'2^{-s})^{2^s}]_{l,y} \\
&\phantom{=}(t' := t/2 ) \\
&\geq \sum_{l\in\statespace_{r}} [S_{r}(t'2^{-s})^{2^s}]_{x,l} [S_{r}(t'2^{-s})^{2^s}]_{l,y} \\
&\phantom{=}(\text{induction}) \\
&=[S_{r}(t2^{-(s+1)})^{2^{(s+1)}}]_{x,y}.
\end{align}
\end{proof}

\begin{proof}[\textbf{Proof of \cref{prop:Skeletoid_err_bound}}]
Note that
\[
[M(t)]_{x,y} - [S(t2^{-k})^{2^k}]_{x,y} = \P_x(X(t)=y) -\P_x(X(t)=y,A_k) = \P_x(X(t)=y,A_k^c),
\]
where the event $A_k$ is defined in \cref{eq:def_Ak}. Summing over $y$,
\begin{align}
\sum_{y\in\statespace} ([M(t)]_{x,y} - [S(t2^{-k})^{2^k}]_{x,y}) = \sum_{y\in\statespace} \P_x(X(t)=y,A_k^c) \leq\P_x(A_k^c)
\end{align}
where the last inequality recognizes the possibility that $Q$ may be non-conservative. Since $\bar{q}>-\infty$, the process can be simulated via uniformization. Let
\[
\tilde N_{l,k} \overset{i.i.d.}{\sim} \text{Poisson}(-\bar{q}t2^{-k}), \hspace{1em} l\in\{0,1,\dots,2^k-1\}
\]
denote the number of uniformized events in each block of the partition. These are independent and identically distributed by definition of uniformization. Also, each of the uniformized events is either a self-transition or a jump. Hence,
\[
\tilde N_{l,k} \le 1 \implies N_{l,k} \le 1.
\]
Then,
\begin{align}
\P_x(A_k^c) &= 1 - \P_x(A_k) \\
&= 1 - \P_x\left(\bigcap_{l=0}^{2^k-1} \{N_{l,k} \le 1\} \right) \\
&\leq 1 - \P_x\left(\bigcap_{l=0}^{2^k-1} \{\tilde N_{l,k} \le 1\} \right) \\
&= 1 - \P_x(\tilde N_{0,k} \le 1)^{2^k} \\
&= 1 - e^{\bar{q}t}\left(1 - \bar{q}t2^{-k} \right)^{2^k}.
\end{align}
Now, use the expansion
\[
e^{-x}\left(1+\frac{x}{n}\right)^n = 1 - \frac{x^2}{2n} + O\left(\frac{x^3}{n^2}\right)
\]
with $x = - \bar{q} t$ and $n = 2^k$ to obtain
\begin{align}
\P_x(A_k^c) &\le 1 - e^{\bar{q}t}\left(1 - \bar{q}t2^{-k} \right)^{2^k} = (\bar{q} t)^22^{-k-1} + O((\bar{q} t)^32^{-2k}).
\end{align}
Since the bound on the right-hand side does not depend on $x$, we conclude that
\[
\|S(t2^{-k})^{2^k} - M(t)\|_\infty \leq \sup_{x\in\statespace} \P_x(A_k^c) \leq (\bar{q} t)^22^{-k-1} + O((\bar{q} t)^32^{-2k}).
\]
\end{proof}

\end{document}